\documentclass[sigconf]{acmart}

\usepackage{graphicx}
\usepackage{xspace}
\usepackage{amsfonts}
\usepackage[ruled,noend,linesnumbered]{algorithm2e}
\usepackage{color}
\usepackage{mathtools}
\usepackage{booktabs}
\usepackage{multirow}
\usepackage{balance}
\usepackage{url}
\usepackage{appendix}
\usepackage{listings}
\usepackage{enumitem}
\usepackage{latexsym}
\usepackage{makecell}
\usepackage{float}
\usepackage{threeparttable}
\usepackage{arydshln}

\newcommand{\ie}{\emph{i.e.,}\xspace}
\newcommand{\eg}{\emph{e.g.,}\xspace}

\newcommand{\resp}{\emph{resp.,}\xspace}
\newcommand{\etal}{\emph{et al.}\xspace}

\newcommand{\eat}[1]{}

\theoremstyle{remark}
\newtheorem{remark}{Remark}

\AtBeginDocument{%
  \providecommand\BibTeX{{%
    \normalfont B\kern-0.5em{\scshape i\kern-0.25em b}\kern-0.8em\TeX}}}

\setcopyright{acmcopyright}
\acmYear{2023}
\acmDOI{10.1145/3626725}\acmPrice{15}

\acmConference[Technical report]{2024 ACM SIGMOD International Conference on Management of data}{June 2024}{Santiago, Chile}
\begin{document}

\title{DP-starJ: A Differential Private Scheme towards Analytical Star-Join Queries}

\author{Congcong Fu}
\author{Hui Li}
\affiliation{
 \institution{Xidian University}
 \city{Xi'an}
 \country{China}
}
\author{Jian Lou}
\affiliation{
\institution{Zhejiang University}
\city{Hangzhou}
\country{China}
}
\author{Jiangtao Cui}
\affiliation{
\institution{Xidian University}
\city{Xi'an}
\country{China}
}
\begin{abstract}
  Star-join query is the fundamental task in data warehouse and has wide applications in On-line Analytical Processing (\textsc{olap}) scenarios. Due to the large number of foreign key constraints and the asymmetric effect in the neighboring instance between the fact and dimension tables, even those latest \textsc{dp} efforts specifically designed for join, if directly applied to star-join query, will suffer from extremely large estimation errors and expensive computational cost.
  
 In this paper, we are thus motivated to propose DP-starJ, a novel \textbf{D}ifferentially \textbf{P}rivate framework for \textbf{star}-\textbf{J}oin queries. DP-starJ consists of a series of strategies tailored to specific features of star-join, including 1) we unveil the different effect of fact and dimension tables on the neighboring database instances, and accordingly revisit the definitions tailored to different cases of star-join; 2) we propose Predicate Mechanism (PM), which utilizes predicate perturbation to inject noise into the join procedure instead of the results; 3) to further boost the robust performance, we propose a \textsc{dp}-compliant star-join algorithm for various types of star-join tasks based on PM. We provide both theoretical analysis and empirical study, which demonstrate the superiority of the proposed methods over the state-of-the-art solutions in terms of accuracy, efficiency, and scalability.
\end{abstract}







\maketitle

\section{Introduction}

\textit{Star-join} query is a common type of query in data warehouse applications, especially on star schema warehouse, where a fact table is joined with one or more dimension tables. It usually performs some filtering on dimension tables, joins the dimension tables with the fact table, and executes some optional aggregation on that. The following provides the formal definition of star-join query.
\begin{definition}[Star-Join]\label{df:11}
Let $\mathbf{R}$ be a database schema containing $n+1$ tables, namely $R_0,...,R_n$. We start with a star-way join :
\begin{equation}
    J := R_0(\mathbf{x}_0) \Join R_1(\mathbf{x}_1) \Join ...\Join R_n(\mathbf{x}_n),
\end{equation}
where $R_1,...,R_n$ are dimension tables and $R_0$ is a fact table. We use $[n]$ to denote $\{1,...,n\}$ and each $\mathbf{x}_i(i \in [n])$ of dimension table $R_i(i\in [n])$ consist of a join key $k_i$ and attribute $a_i$, $\mathbf{x}_i =\{k_i, a_i\}$. Yet $\mathbf{x}_0$ consist of all join keys $k_i$ and a measure attribute $a_0$, $\mathbf{x}_0 =\{k_1,..., k_n, a_0\}$. Let $var(J) := a_0 \cup ... \cup a_n$ is a set of variables in the join result $J$. Each attribute $a_i$ has a finite domain $dom(a_i)$ with size $|dom(a_i)| = m_i$, the full domain of $\mathbf{R}$ is $dom(\mathbf{R}) = dom(a_0) \times ... \times dom(a_n)$ and has size $m = \prod_i m_i$.    
\end{definition}

\begin{figure}[t]
  \centering
  \includegraphics[width=0.8\linewidth, height = 4.2cm]{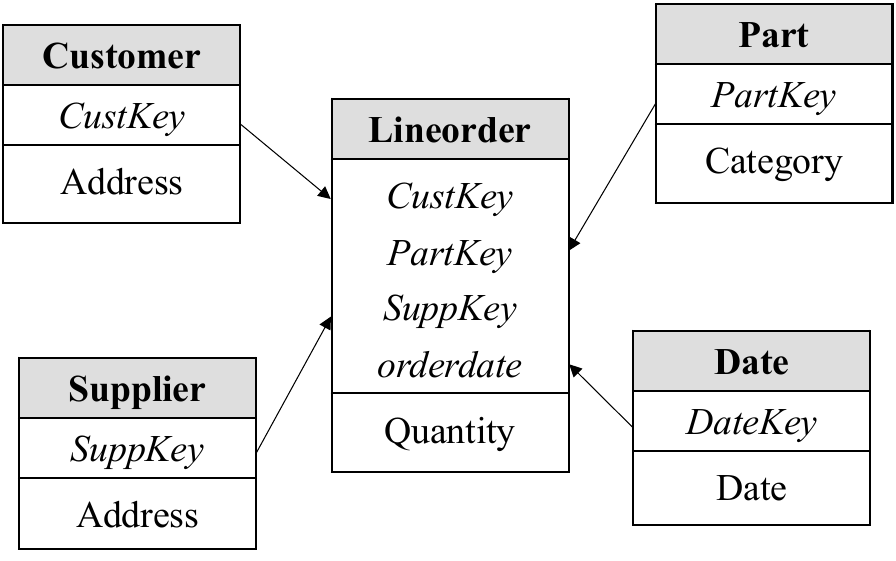}\vspace{-2ex}
  \caption{An example of star schema with 4 dimension tables}\label{fig_0}\vspace{-3ex}
\end{figure}

\begin{example}\label{ex1}
Suppose a data analyst is interested in the total number of items sold in the first half of this year in a given region, s/he would execute the following query (assuming the query is performed on the Star Schema 
 Benchmark (\textsc{ssb})~\cite{o2007star}): 
\begin{lstlisting}[language=SQL,keywordstyle=\color{blue},mathescape,basicstyle=\ttfamily,escapeinside=\{\},showstringspaces=false]
SELECT count($*$)
FROM {Date}, Customer, Supplier, Part, Lineorder 
WHERE Lineroder.CK = Customer.CK 
    AND Lineroder.SK = Supplier.SK
    AND Lineroder.PK = Part.PK 
    AND Lineorder.orderdate = {Date}.DK
    AND Customer.region = '[REGION]'
    AND Supplier.region = '[REGION]'
    AND {Date}.{month} < 7;  
\end{lstlisting} 
\end{example}
Figure~\ref{fig_0} shows an example of standard star schema where the query in Definition~\ref{df:11} can be applied and Example~\ref{ex1} shows an example of a star-join query. Such queries involving star-joins are ubiquitous within analytical tasks and act as a core query category in the data warehouse. Due to such a pivotal role in data analytics,  star-join has been extensively studied in the literature and widely applied in On-line Analytical Processing (\textsc{olap}) practice. Meanwhile, unlike the other types of joint queries that \emph{all tables} can be connected to each other, \emph{all dimension tables} in star-join will be directly linked to the fact table through the foreign-key constraints. In the above example, the relations touched by the query contain private information, \eg customer $c_1$ has placed a particular order $o_2$, suppliers $s_1,\ldots,s_m$ provide an item $i_1$, of which the privacy must be protected in practical scenarios.

At present, differential privacy (\textsc{dp}) has become a popular solution in privacy-preserving data analytics as it provides a statistically rigorous privacy guarantee. Since its introduction \cite{dwork2006differential,dwork2006calibrating}, \textsc{dp} has attracted ever-growing interest in academia, government agencies, and industry. The standard \textsc{dp} mechanism (e.g., \textit{Laplace Mechanism}) first finds the global sensitivity of the query, then it adds a carefully calibrated random noise tailored to the query result. 
High sensitivity can introduce large noise, which results in a distorted query result offering poor utility. In particular, the global sensitivity of the query refers to how much the query result may change in two neighboring instances of databases. Consequently, a proper definition of neighboring instances is of great importance in \textsc{dp}, which not only determines whether the \textsc{dp} mechanism built on it offers sound and practical privacy protection, but also affects the sensitivity and eventually the utility of the \textsc{dp} mechanism.
Since such noise for the privacy-preserving purpose will unavoidably cause utility degradation for the query result, a central problem in \textsc{dp} is how to achieve a satisfactory trade-off between privacy and utility. 
Existing works \cite{kasiviswanathan2013analyzing, wang2019locally,zheng2022secure,xu2013differentially,liew2022network,qardaji2013understanding,takagi2021p3gm} have proven that \textsc{dp} mechanism usually achieves a better privacy-utility trade-off when its design is tailored to the specific data analysis task under consideration. 
In this regard, this paper, for the first time, proposes a solution towards answering star-join queries under differential privacy. 

Roughly, the efforts of recent works in \textsc{dp} query processing~\cite{dong2021residual, kotsogiannis2019privatesql, tao2020computing} focus on three aspects to reduce the high global sensitivity: reduce query sensitivity by utilizing the upper bound of the local sensitivity, design an algorithm that effectively computes tight local sensitivity, and transform the database instance by deleting some tuples that are highly sensitive. However, \textit{different from other types of queries that all tables can be linked through join operations, the star-join query has a non-trivial number of foreign key constraints that a single fact table references a series of dimension ones}. Due to that, answering a star-join query in a \textsc{dp} manner is more challenging because high global sensitivity results from the large number of foreign key constraints in star-join query, the \textsc{dp} mechanism of high global sensitivity fails to work, as the output of a join may contain duplicate sensitive rows. This duplication is difficult to bound as it depends on the join type, join condition, and the underlying data. Therefore, the global sensitivity becomes unbounded when joins are present because a single tuple may affect many join results. Therefore, the existing \textsc{dp}-compliant query strategies with a trusted server may not be able to provide satisfactory \textit{utility} and \textit{efficiency}, which motivates us to present the solutions in this work.
\begin{example}\label{ex2}
The following is a simplest star-join query: \\ $q:= Customer(\underline{CK},Address,...) \Join Lineorder(CK,orderdate,...)$. 
\end{example}
Here, $Customer$ may store customer information and $Lineorder$ contains the orders the customers have placed. Then this query simply returns the total numbers of orders. Suppose the identities for the entities in $Customer$ are private information we aim to protect. Unfortunately, the global sensitivity of this query is $\infty$ under existing DP solutions~\cite{dwork2006calibrating}. The reason is as follows, a customer could have an unbounded number of orders, and adding such a customer to the database can cause an unbounded change in the query result theoretically. To address this issue, some works~\cite{nissim2007smooth, dong2021residual} suggest adding data-dependent noise to the query result. For instance, \cite{nissim2007smooth} proposes to use the local sensitivity, \ie the sensitivity of the join query on the given database instance, which is usually much lower than global sensitivity. However, if applied in star-join query, it still leads to high sensitivity and further results in a low utility. The key challenge is how to decrease the global sensitivity of the star-join queries when designing the \textsc{dp} schemes. 

Meanwhile, within star-join the tuples from the fact and the dimension tables shall affect the query result differently. For this reason, there also exist several different cases for neighboring database instances depending on whether the fact or the dimension tables are private. Therefore, before presenting a well-designed \textsc{dp}-compliant star-join solution, \textit{it is necessary to revisit the definition of neighboring database instances due to the asymmetry between the fact and dimension table}. Accordingly, the \textsc{dp}-compliant star-join solution should take into account the fact that the definition of neighboring database instances may vary between scenarios.

In this paper, we systematically investigate the differential privacy star-join query problem. Our study first reveals that the existing approaches of the traditional \textsc{dp}-compliant join schemes~\cite{dong2021residual,tao2020computing,johnson2018towards}, which work by adding subtly noise to the join \textit{result}, fail to achieve a satisfactory utility and efficiency in star-join queries. 
Thus, we are further motivated to propose an advanced approach called DP-starJ, a \textbf{D}ifferentially \textbf{P}rivate framework for \textbf{star}-\textbf{J}oin queries. To achieve that, we first investigate and unveil the asymmetry between the fact and dimension tables in the effect on neighboring instances of star-join. Driven by that unique nature, instead of considering a uniform definition and simplified case for neighboring database instance as existing \textsc{dp} schemes~\cite{dong2022r2t, dong2021residual}, we propose a fine-grained definition for neighboring database instance tailored to the asymmetry characteristics of star-join task. On the other hand, as discussed above, truncating some highly sensitivity tuples or adding data-dependent noise towards the result fails to achieve a satisfactory utility and efficiency due to the large number of foreign keys, we are also motivated to propose a new perturbation mechanism to achieve superior utility and efficiency by adding the data-independent noise with bounded global sensitivity, namely Predicate Mechanism. Using the proposed mechanism as a building block, we present an \textsc{dp}-compliant star-join algorithm for various types of star-join tasks (\ie aggregate query,  ``group\_by'' operation, and workload queries). Further, our theoretical study shows that the proposed methods obtain asymptotically optimal error bound on star-join. Empirical study over several real-world datasets justifies the superiority of our solution in the aspect of both utility and efficiency across various star-join tasks.

The contributions of this paper are summarized as follows:
\begin{itemize}
    \item We unveil the asymmetry between the fact and dimension tables in affecting the neighboring database instances and revisit the accordingly definitions to tailor to different cases of star-join. 
    \item We propose a Predicate Mechanism under DP-starJ, which designs a new perturbation strategy to inject noise towards the star-join procedure instead of purely the results. Meanwhile, we further propose an \textsc{dp}-compliant star-join algorithm for various types of star-join tasks.
    \item We prove theoretically that the proposed method obtains asymptotically optimal error bound on star-join queries and experimental study justifies the superiority of our solution in the aspects of both utility and efficiency.
\end{itemize}

\vspace{-2ex}\section{Related Work}
Early works mostly focus on answering a given arbitrary SQL query under \textsc{dp}, which is acknowledged as the holy grail of private query processing. There have been several works on answering various types of queries under \textsc{dp} \cite{barak2007privacy,narayan2012djoin, zeighami2021neural, cormode2012differentially, kato2022hdpview} but not star-join, which has always been the core and basis for the majority of \textsc{olap} applications~\cite{galindo2008optimizing}. At present, there is no work that is specifically designed to answer the star-join query in a privacy-preserving manner under trusted server settings. Since the elegant work by Dwork \cite{dwork2006differential}, there are plenty of works \cite{tao2020computing, dong2021residual,dong2022nearly,cai2023privlava} proposed to limit the sensitivity of join queries and extensions for optimizing multi-join queries. In addition, \textsc{dp}-compliant SQL query processing has also been extensively applied in industrial systems, for instance, Uber implements Flex \cite{johnson2018towards} that answers SQL queries with \textsc{dp}. 

Many technologies have been proposed to answer set counting queries over a single relation with different predicates \cite{barak2007privacy, blasiok2019towards, day2016publishing, hardt2012simple, nikolov2013geometry, xiao2010differential, qardaji2013understanding, zhang2014towards, qardaji2014priview}. Most existing work on join queries can only support restricted types of joins, such as joins with primary keys \cite{arapinis2016sensitivity, mcsherry2009privacy, narayan2012djoin, palamidessi2012differential, proserpio2012calibrating} and joins with a fixed join attribute \cite{wilson2020differentially}. One approach is to reduce the high sensitivity of join queries by truncation. For instance, PrivateSQL \cite{kotsogiannis2019privatesql} uses naive truncation to truncate the tuples with high degrees. Tao \etal \cite{tao2020computing} use naive truncation to truncate the tuples with high sensitivity for some queries without self-join and they propose a mechanism to select the truncation threshold. Dong \etal \cite{tao2020computing} proposed a mechanism Race-to-the-Top (R2T), which can be used to adaptively choose the truncation threshold. However, if applied in star-join, the truncation-based solution will cause a significant biased result due to the foreign key constraints between the large number of dimension tables and fact table. Another approach is adding data-dependent noise calibrated by other types of sensitivity rather than global sensitivity. Smooth sensitivity~\cite{nissim2007smooth} is a popular approach for dealing with multi-way joins. Elastic sensitivity~\cite{johnson2018towards} and residual sensitivity~\cite{dong2021residual}, both of which are efficiently computable versions of smooth sensitivity, can handle join queries efficiently. However, smooth sensitivity (including any efficiently computable version) cannot support foreign key constraints, which are important to model the relationship between an individual and all his/her associated records in a relational database. Similarly, these methods cannot balance utility and efficiency in the star-join query under \textsc{dp}.

In comparison, the star-join queries have \textit{a non-trivial number of foreign key constraints} in multi-way joins scenarios and the goal is to effectively get accurate query answers even when the star-join query contains a large number of dimension tables. The existing \textsc{dp}-compliant query strategies with trusted servers do not satisfy the practical requirements, which motivates us to present the solutions in this work.

\vspace{-2ex}\section{Preliminaries and Problem Definition}
\subsection{Preliminaries}
\noindent{\bfseries Star-join query} 

Many relational data warehouse designs today follow a so-called dimensional modeling approach that has been made popular by Galindo \etal \cite{galindo2008optimizing}. Dimensional modeling relies on the distinction of dimension tables with relatively static information in contrast to fact tables that store transactional statistical information. For instance, according to the TPC-H benchmark schema, dimensional table hold master data representing entitles such as \textit{part}, \textit{customers}, \textit{suppliers}, and \textit{date}. In comparison, the fact table in turn stores transactional data, \eg \textit{lineorder} contains statistics about sales or orders. Dimension tables and fact tables are correlated with each other by foreign key constraints. Usually, fact tables are several orders of magnitude larger than the dimension ones. Dimensional modeling leads to the well-known so-called star schema and snowflake schema design for data warehousing. A star schema consists of a fact table in the center of the star, and it is very popular for modeling data warehouses and data marts. The fact table contains foreign keys, which are pointing to the dimension tables, and the dimension tables contain a key used to joining with the fact table and additional attributes. 


Star-join queries are queries on a database instance that the fact table is joined with one or more dimension tables, it selects several measures of interest from the fact table, joins the fact rows with one or several dimensions with respect to the keys, places filter predicates on the business columns of the dimension tables, performs grouping if required, and finally aggregates the measures retrieved from the fact table. As the star-join query in \textsc{olap} task places filter predicates on the attributes of the dimension tables, and finally aggregates the measure attribute from the fact table, the star-join query can be converted into a predicate query. Predicate queries are a versatile class, consisting of queries that satisfying any logical predicate. A predicate corresponds to a condition in the \textsf{WHERE} clause of an SQL statement, and a star-join query is a SQL query with aggregation on measure attributes of the fact table and a predicates with equality and range constraints on some dimension tables. The following showcase the template for star-join queries in the form of a standard predicated \textsf{SELECT} SQL statement: \begin{lstlisting}[language=SQL,keywordstyle=\color{blue},mathescape,escapeinside=\{\},showstringspaces=false,xleftmargin=.05\textwidth]
SELECT Aggr(*) FROM R WHERE $\Phi$;
\end{lstlisting}   
\textbf{Aggr($\ast$)} refers to an aggregate function (\eg COUNT, AVG, SUM) over the fact table. $\Phi$ means conjunctions of filter conditions that consists of arbitrary predicates $\phi$ on attributes over the dimension tables. When a star-join query refers only to an attribute $a_i(i \in [n])$ in dimension table $R_i$ we may say that it is defined with respect to $a_i$ and annotate it as $\Phi := \phi_{a_i}, (\phi_{a_i} : dom(a_i) \rightarrow \{0,1\})$. Similarly, if $\phi_{a_i}$ and $\phi_{a_j}$ are predicates on dimension tables $R_i$ and $R_j$ in a star-join query, then $\Phi$ is the conjunctions of predicates $\Phi := \phi_{a_i} \wedge \phi_{a_j}, (\phi_{a_i} \wedge \phi_{a_j} : dom(a_i \cup a_j) \rightarrow \{0,1\})$.

Let $\mathbf{D}_s$ be a database instance over star schema and a star-join query $Q$ aggregates over the join result $J(\mathbf{D}_s)$. More abstractly, let $\Phi : dom(var(J)) \rightarrow \{0,1\}$ be an indicator function and the join result satisfy the filter predicate, and $\mathbf{w}(t)$ assigns a non-negative integer weight to the join results only depending on the tuple $t$. Given the above, we denote the query result of $Q$ on $\mathbf{D}_s$ as $Q(\mathbf{D}_s)$, which can be formally represented as follows.
\begin{equation}\label{eq1}
    Q(\mathbf{D}_s) = \sum_{t \in J(\mathbf{D}_s)} \Phi(t) \cdot \mathbf{w}(t)
\end{equation}
Note that the function $\Phi$ only depends on the star-join query and $t$ is the tuple in join result $J(\mathbf{D}_s)$. In addition, a star-join query with arbitrary predicate over $var(J)$ can be easily incorporated into this formulation (boolean function): If some $t \in J(\mathbf{D}_s)$ does not satisfy the predicate, we simply set $\Phi(t) = 0$. For a counting query, \textbf{Aggr($\ast$)} will appear in the form of a \textsf{COUNT} function, $\mathbf{w}(\cdot) = 1$; for other aggregation query, \eg SUM($a_0$), \textbf{Aggr($\ast$)} refers to a \textsf{SUM} function, $\mathbf{w}(t)$ is the value of attribute $a_0$ for $t$.

Consider the star-join query towards a database instance in Example~\ref{ex1}. The query consists of a set of single-table predicates as follows: in the \textit{Date} table, define predicate $\phi_{Date} = \{\mathbb{I}[t_{month} < 7] \}$, $\phi_{Supp} = \{\mathbb{I}[t_{region}  =  {\rm REGION}]\}$ and $\phi_{Cust} = \{\mathbb{I}[t_{region}  =  {\rm REGION}] \}$ in \textit{Supplier} and \textit{Customer} tables. The composite predicate for the query can be expressed as the product $\Phi : \phi_{Date} \wedge \phi_{Cust} \wedge \phi_{Supp}$.


\noindent{\bfseries Differential Privacy in Relational Databases with Join Query}
Differential Privacy (\textsc{dp}) provides a rigorous privacy guarantee, which has become the de facto privacy-preserving notion in many applications. Before presenting the formal definition of \textsc{dp}, we shall introduce the notion of the neighboring database first. For two database instances $\mathbf{D}$ and $\mathbf{D}'$, the distance between $\mathbf{D}$ and $\mathbf{D}'$, denoted $d(\mathbf{D},\mathbf{D}')$, is the minimum number of steps on which they differ. If $d(\mathbf{D},\mathbf{D}')=1$, we call $\mathbf{D}$, $\mathbf{D}'$ neighboring database instances.

\begin{definition}[Differential Privacy] 
	A randomized algorithm $\mathcal{A}$ satisfies $(\epsilon,\delta)$ - differential privacy, where $\epsilon, \delta > 0$, if for any pair of neighboring databases $\mathbf{D}$, $\mathbf{D}'$  and any output range $\mathcal{S} \subseteq  Range(\mathcal{A})$,
  \begin{equation}
    Pr[\mathcal{A}(\mathbf{D}) \in \mathcal{S}]\leqslant e{^\epsilon}\cdot Pr[\mathcal{A}(\mathbf{D}') \in \mathcal{S}] + \delta,
  \end{equation} 
  where the probability is taken over the randomness of $\mathcal{A}$. When $\delta = 0$, it is referred to as pure differential privacy, the algorithm $\mathcal{A}$ satisfies $\epsilon$ - differential privacy.
\end{definition}
In the above definition, $\epsilon$ refers to the privacy budget, which directly restricts the degree of the privacy protection of the algorithm $\mathcal{A}$. Typically, a smaller value of $\epsilon$ corresponds to a stronger privacy guarantee. In addition, $\delta$ should be much smaller than $1/N_\mathbf{D}$ to ensure the privacy of each individual record, where $N_\mathbf{D}$ refers to the size of the database instance.

Differential privacy is usually achieved by adding random noise drawn from a certain zero-mean probability distribution to the query result. Notably, the magnitude of the random perturbation positively correlates with the difference between the query results on $\mathbf{D}$ and $\mathbf{D}'$, which refers to the notion of sensitivity. The most basic framework for achieving differential privacy is the Laplace mechanism, and the noise is scaled according to the \textit{global sensitivity} of the query, defined as follows.
\begin{theorem}[Laplace Mechanism]\label{theo:41}
  The algorithm $\mathcal{A}(\mathbf{D}) = Q(\mathbf{D}) + Lap(\frac{GS_Q}{\epsilon})$ is $\epsilon$ - differential privacy.
\end{theorem}
\begin{definition}[Global Sensitivity]
  Let $Q$ denote a particular query, then the global sensitivity of $Q$, denoted $GS_Q$, is
  \begin{equation}
    GS_Q =  \max \limits_{\mathbf{D},\mathbf{D}',d(\mathbf{D},\mathbf{D}')=1}\parallel Q(\mathbf{D})- Q(\mathbf{D}')\parallel.
  \end{equation} 
\end{definition}
The global sensitivity of the query is defined as the maximal $L_1$-norm distance between the exact answers of the query $Q$ on any neighboring databases $\mathbf{D}$ and $\mathbf{D}'$. However, unfortunately, the global sensitivity of many queries can be very high. What is worse, for the join operator the global sensitivity can be unbounded. Nissim \etal ~\cite{nissim2007smooth} proposed a local measure of sensitivity:
\begin{definition}[Local Sensitivity]
  For a query $Q$, the local sensitivity of $Q$ given the database instance $\mathbf{D}$, denoted as $LS_Q(\mathbf{D})$ is as follows:
  \begin{equation}
    LS_Q(\mathbf{D}) =  \max \limits_{\mathbf{D}',d(\mathbf{D},\mathbf{D}')=1}\parallel Q(\mathbf{D})- Q(\mathbf{D}')\parallel.
  \end{equation} 
  where the maximum is taken over all neighbors $\mathbf{D}'$ of the particular instance $\mathbf{D}$.
\end{definition}
 Note that, $GS_Q = \max \limits_{\mathbf{D}} LS_Q(\mathbf{D})$. The local sensitivity is much smaller than global sensitivity in most real-world scenarios. However, an algorithm that releases query results with noise scale proportional to $LS_Q(\mathbf{D})$ on instance $\mathbf{D}$ may not satisfy differential privacy, since $LS_Q(\mathbf{D})$ and $LS_Q(\mathbf{D}')$ can differ a lot on two neighboring instances $D$ and $D'$. Large differences in the amounts of noise added to $Q(\mathbf{D})$ and $Q(\mathbf{D}')$ may leak sensitive information. To address the issue, Nissim \etal ~\cite{nissim2007smooth} proposed the approach that selects noise magnitude according to a smooth upper bound on the local sensitivity instead of using the local sensitivity itself. But differently, compared with the local sensitivity, it is the maximum local sensitivity attained among neighboring instances, the tightest bound is called the \textit{smooth sensitivity}. The smooth sensitivity is based on the \textit{local sensitivity at distance} $t$, \ie $LS_Q^{(t)}(\mathbf{D})$, which is defined as
\begin{equation}
    LS_Q^{(t)}(\mathbf{D}) =  \max \limits_{{\mathbf{D}',d(\mathbf{D},\mathbf{D}') \leq t}}LS_Q(\mathbf{D}').
\end{equation} 

\begin{definition}[Smooth Sensitivity]
  The $\beta$ - smooth sensitivity of $Q$, denoted $SS_Q(\mathbf{D})$, is
  \begin{equation}
    SS_Q(\mathbf{D}) =  \max \limits_{t \geq 0}e^{-\beta t}LS_Q^{(t)}(\mathbf{D}).
  \end{equation} 
\end{definition}
$SS_Q(\mathbf{D})$ and $SS_Q(\mathbf{D}')$ differ by at most a constant factor on any two neighboring instances $\mathbf{D}$ and $\mathbf{D}'$ to ensure the ``smoothness'' of $SS_Q(\cdot)$, and the level of smoothness is parameterized by a value $\beta$ (a smaller value leads to a smooth bound) that depends on $\epsilon$.

\subsection{Problem Definition}

\noindent{\bfseries Differential Privacy in Star-join query}

Star-join queries are the most prevalent kind of queries in data warehousing, \textsc{olap} and business intelligence applications. Hence, answering star-join query under differential privacy can definitely benefit wide applications in privacy-preserving tasks in the \textsc{olap} scenarios. Therefore, in this work, we aim to propose the first \textsc{dp}-compliant star-join solution. However, before presenting the solution, due to the special characteristics of the query, it is necessary to reconsider the definition of differential privacy of star-join query. In this subsection, we introduce differential privacy in the star-join query, including neighboring database instances in different situations (fact table and dimension table), and differential privacy in single private relation and multi-private relations with star-join query afterwards.

Consider a database instance $\mathbf{D}_s$ over star schema $\mathbf{R}:\{ R_0, R_1, ..., R_n \}$, where $R_0$ is a fact table and the rest are $n$ dimension tables. Given a star-join query $Q$ shown in Definition~\ref{df:11}, let $N=|\mathbf{D}_s|$ be the input size, and denote the result of $Q$ on $\mathbf{D}_s$ as $Q(\mathbf{D}_s)$. We consider a \textsc{dp}-compliant star-join based on neighboring instances $\mathbf{D}_s$, $\mathbf{D}'_s$.
\begin{definition}[Differential Privacy in Star-Join Query] \label{df:36}
	A randomized mechanism $\mathcal{A}$ satisfies $\epsilon$ - differential privacy if for neighboring instances $\mathbf{D}_s$, $\mathbf{D}'_s$ over star-join, where $\epsilon > 0$, and any output range $\mathcal{S} \subseteq  Range(\mathcal{A})$, 
  \begin{equation}
    Pr[\mathcal{A}(\mathbf{D}_s) \in \mathcal{S}]\leqslant e{^\epsilon}\cdot Pr[\mathcal{A}(\mathbf{D}'_s) \in \mathcal{S}],
  \end{equation} 
  where the probability is taken over the randomness of $\mathcal{A}$.
\end{definition}
In the above definition, neighboring instances $\mathbf{D}_s$, $\mathbf{D}'_s$ over star schema should differ by one tuple according to the notion of the neighboring database. However, in the star schema, each dimension table is independent of each other and has a foreign key constraint referenced by the fact table. The tuples in the fact table and dimension tables exert different effects on the query result due to the asymmetric characteristics for both types of tables within the star-join procedure. Therefore, it is necessary to revisit the definition for neighboring instances $\mathbf{D}_s$, $\mathbf{D}'_s$. At the same time, database instances may contain a single private relation or multi-private relations in practical applications. Based on the above reasons, in this subsection, we consider the following situations of neighboring database instances $\mathbf{D}_s$, $\mathbf{D}'_s$.

\noindent\textbf{Scenario-dependent Neighboring Database Instance.} 

As we have discussed above, the unique characteristics of star-join rely on the fact that there exist a large number of foreign key constraints between the fact and dimension tables. As a result, the difference in a single tuple within a dimension table may result in a group of different tuples in the fact one. Hence, the asymmetry between both types of tables leads to different scenarios for neighboring instances. In the following, we shall discuss them accordingly.

\begin{definition}[$(a,b)$-private]
Given the star-join task shown in Definition~\ref{df:11}, which contains at least one sensitive table, we refer to the scenario as \textbf{$(a,b)$-private} if a number of $a$ ($a\in\{0,1\}$) fact tables and $b$ ($b\le n, a+b\ge 1$) dimension ones are sensitive.
\end{definition}

(1)\textit{$(0,k)$-private}. The private relations are all dimension tables, $R^1_p,...,R^k_p \in \{R_i\}(i\in [n], k \leq n)$. When the database instance $\mathbf{D}_s$ exists the foreign key constraint that table has foreign key referencing the primary key (PK) of the other table, the two instances $\mathbf{D}_s$ and $\mathbf{D}'_s$ are considered as neighbors if $\mathbf{D}'_s$ can be obtained from $\mathbf{D}_s$ by: deleting a tuple $t$ from the referenced table, and a set of tuples that reference $t$ in the referencing table. As each dimension table has a foreign key constraint with the fact table, we adopt the DP policy in star-join query, which defines neighboring instances by taking foreign key constraints into consideration. The basic private relation of $(0,k)$-private is to only include one dimension table, that is, when $k = 1$. Therefore, we refer to $\mathbf{D}_s$, $\mathbf{D}'_s$ as neighboring instances over star schema if all tuples in the difference between $\mathbf{D}_s$ and $\mathbf{D}'_s$ reference a single tuple $t_p$ in the private dimension table $R^k_p$. In particular, $t_p (t_p \in R^k_p)$ may also be deleted, in which case all tuples referencing $t_p$ in the fact table must be deleted in order to preserve the foreign key constraints. When $k > 1$, since each dimension table is independent of each other and the fact table has foreign keys referencing the primary key of each dimension one, thus we assign unique identifiers to the conjunction of all foreign keys in the fact table. If $\mathbf{D}'_s$ can be obtained from $\mathbf{D}_s$ by deleting a tuple $t^i_p \in R^i_p$ for each private relations, as well as all the tuples in the fact table referencing the same tuple $t \in t^1_p(PK) \wedge ... \wedge t^k_p(PK)$, we call $\mathbf{D}_s$, $\mathbf{D}'_s$ neighboring instances in this case.

(2)\textit{$(1,k)$-private}. The private relations contains the fact table. The simplest scenario of $(1,k)$-private is $k = 0$, which means that only the fact table is private. When $k=0$, two instances can only differ at one tuple in the fact table, $\mathbf{D_s}$, $\mathbf{D}'_s$ are referred to as neighboring instances, $d(\mathbf{D}_s, \mathbf{D}'_s) = 1$. Another scenario of $(1,k)$-private is the case when $k \neq 0$, \ie some of the dimension tables are private. In this case, two neighboring instances $\mathbf{D}_s$, $\mathbf{D}'_s$, can differ at one tuple in the fact table. Moreover, similar to $(0,k)$-private, $\mathbf{D}'_s$ also needs to be obtained from $\mathbf{D}_s$ by deleting a tuple $t^i_p \in R^i_p$ from each private dimension tables, as well as all the tuples in the fact table referencing the same tuple $t \in t^1_p(PK) \wedge ... \wedge t^k_p(PK)$.

The above outlines the different cases for neighboring instances $\mathbf{D}_s$, $\mathbf{D}'_s$ in the star-join query. In Definition~\ref{df:11}, star-join queries are transformed to predicate queries in the multidimensional data cube. Therefore, the algorithm that satisfies differential privacy is implemented for each predicate constraint of the star-join query $Q$, so that the query $Q$ conforms to differential privacy.

\section{Basic Mechanism for Star-join Query: Output Perturbation}\label{sec3}
In order to systematically find the ideal solution for answering star-join query under \textsc{dp}, we investigate ways through both the output and input perturbations. In this section, we propose the basic approach for \textsc{dp}-compliant star-join query by a pair of output-based perturbation mechanisms. Aside from that, we also conduct a theoretical utility study, which shows that the basic mechanism achieves a satisfactory (although not elegant) trade-off between utility, efficiency, and scalability.

Intuitively, following the standard \textsc{dp} solutions, we can propose a basic strategy by approximating real-valued functions based on adding a small amount of random noise to the true answer. In particular, we introduce both a data-independent approach and a data-dependent one to the star-join query result according to whether the global sensitivity of star-join query is bounded. In a data-independent approach, if the global sensitivity of star-join query is bounded, the server is in charge of adding random noise to the query result. The most popular method is to rely on the Laplace Mechanism (LM) that scales according to the global sensitivity $GS_Q$ of the star-join query $Q$. The variance of the Laplace Mechanism is $2(\frac{GS_Q}{\epsilon})^2$.

In star-join query, this method is only applicable with the $(1,0)$-private scenario, where the fact table is the only one that is sensitive. Besides that, the Laplace mechanism will fail to work in the $(\cdot,k)$-private relation contains dimension table due to the unbounded global sensitivity. Notably, in practical scenarios, sensitive information is mostly contained in the dimension tables rather than the fact one (\eg $Customer$ is a private relation that needs to be protected). 

For those cases when private relation includes dimension table, that is, the global sensitivity is unbounded, we first consider to adopting a data-independent approach by utilizing the Truncation Mechanism (TM) that bounds the global sensitivity by simply deleting all records, the sensitivity of which is larger than a predefined threshold $\tau$, before adding random noise to the true answer. However, a well-known limitation of the truncation mechanism is the bias-variance trade-off: a large threshold $\tau$ will lead to large random noise with tremendous variance; while a small $\tau$ may introduces a bias as large as the query result itself. When the private relation contains dimension table, due to the aforementioned limitation of LM and TM in the data-independent approach, we select to adopt a data-dependent approach by injecting data-dependent noise into the query result. 

The data-dependent approach involves applying Local Sensitivity (LS) and Race-to-the-Top (R2T) to the star-join query. The LS is usually a two-phase strategy as follows.
\begin{enumerate}
    \item compute the upper bound of local sensitivity $\hat{LS}_Q(\mathbf{D}_s)$ in star-join query $Q$ with database instance $\mathbf{D}_s$;
    \item add the noise that calibrates the size of $\hat{LS}_Q(\mathbf{D}_s)$ to the query result. 
\end{enumerate}
In general, there are two mechanisms for implementing LS, namely Cauchy Mechanism and Laplace Mechanism. Cauchy Mechanism works by setting $\beta = \frac{\epsilon}{2(\gamma + 1)}$, and then adds noise $Cauchy(\frac{\hat{LS}_Q(\mathbf{D}_s)}{\beta})$ to the query answer $Q(\mathbf{D}_s)$. It preserves $\epsilon$ - differential privacy, where $Cauchy(\cdot)$ is drawn from the general Cauchy distribution. For instance, suppose we set $\gamma = 4$ for which $Var(Cauchy(\cdot)) =1$, and the noise level of Cauchy Mechanism is thus $(\frac{10\hat{LS}_Q(\mathbf{D}_s)}{\epsilon})^2$. Notably, as there is a long tail in the general Cauchy distribution, which decays only polynomially compared with the exponential decay within the Laplace distribution, one can use the Laplace distribution to achieve a better concentration. However, the Laplace Mechanism only yields $(\epsilon, \delta)$ - differential privacy. The Laplace Mechanism works by setting $\beta = \frac{\epsilon}{2ln(\frac{2}{\delta})}$, and adds noise $Lap(\frac{2\hat{LS}_Q(\mathbf{D}_s)}{\epsilon})$ to the true answer $Q(\mathbf{D}_s)$. Since $Var(Lap(\cdot)) = 2$, the noise level of Laplace Mechanism is $8(\frac{\hat{LS}_Q(\mathbf{D}_s)}{\epsilon})^2$.

Another method in the data-dependent approach is Race-to-the-Top (R2T). It is a truncation mechanism with foreign key constraints in join query, and can be used in combination with any truncation method. The basic idea of R2T is to try geometrically increasing values of truncation threshold $\tau$ and somehow pick the “winner” from all the trials. The R2T first computes the query result $Q(\mathbf{D}_s, \tau)$ with various threshold $\tau$, and then adds $Lap(\frac{\tau}{\epsilon})$ to $Q(\mathbf{D}_s, \tau)$ to get the noise result $\hat{Q}(\mathbf{D}_s, \tau)$, which would turn it into an $\epsilon$ - differential privacy mechanism. Finally, returning the maximum $\hat{Q}(\mathbf{D}_s, \tau)$ preserves \textsc{dp} by the post-processing property of differential privacy. The R2T works as follows: 

For $\tau^{(j)}, j = 1,...,log(GS_Q)$, 
\begin{equation}
\begin{split}
    \hat{Q}(\mathbf{D}_s, \tau^{(j)}) = Q(\mathbf{D}_s, \tau^{(j)}) + Lap(log(GS_Q)\frac{\tau^{(j)}}{\epsilon}) \\ - log(GS_Q)ln(\frac{log(GS_Q)}{\alpha}) \cdot \frac{\tau^{(j)}}{\epsilon},
\end{split}
\end{equation} 
and then outputs $max\{max_j\hat{Q}(\mathbf{D}_s, \tau^{(j)}), Q(\mathbf{D}_s, 0)\}$, where $\alpha$ is the probability concern about the utility. The R2T mechanism satisfies $\epsilon$ - differential privacy by the basic composition theorem~\cite{dwork2014algorithmic}. Note that, $Q(\mathbf{D}_s, \tau)$ is different in queries with and without self-join, it may rely on Linear Program(LP)-based truncation mechanism when there exists self-join in the query. For the utility of R2T, we have $Q(\mathbf{D}_s) - 4log(GS_Q)ln(\frac{log(GS_Q)}{\alpha})\frac{\tau^{*}(\mathbf{D}_s)}{\epsilon} \leq \hat{Q}(\mathbf{D}_s)$ with probability at least $1 - \alpha$. Hereby $\tau^{*}(\mathbf{D}_s)$ means a bound of threshold that holds for any $\tau \geq \tau^{*}(\mathbf{D}_s)$, $Q(\mathbf{D}_s, \tau) = Q(\mathbf{D}_s)$.

\begin{remark}
In the star-join query, the sensitivity of the query plays an important role in the output mechanism. From the aspect of the output perturbation, the utility is directly affected by the noise that is scaled according to the sensitivity of the star-join query. Both the global and the local sensitivity are extremely high, due to the existence of join operations in star-join query. Especially for an $n$-way star join, the global sensitivity can be as high as $O(N^{n-1})$, which is unbounded as $N=|\mathbf{D}_s|$ is the input size. Therefore, this brings down the utility because the $GS_Q$ of the star-join query can be $\infty$ under pure \textsc{dp}. Although using the instance-depended noise, the output mechanism has the intrinsic limitation on achieving high utility due to the fact that the smooth upper bound of $LS_Q(\mathbf{D}_s)$ is very large in practical applications, and the computational cost of it is extremely high. In short, the high sensitivity of star-join query results extremely limits the utility level that the basic output perturbation mechanism can achieve.    
\end{remark}
\begin{remark}
Although the output mechanism adopts the smooth sensitivity to reduce the noise for better utility, in fact, it is shown that for certain problems, computing or even approximating the smooth sensitivity is NP-hard~\cite{tao2020computing}. Therefore, the computational hardness of the smooth sensitivity of star-join queries increases with the increase of multi-way joins. \cite{dong2021residual} argues that it may not be NP-hard, and even if there is a polynomial-time algorithm to compute the smooth sensitivity, it will be inevitably too complicated in practice. Thus it is challenging for the output perturbation mechanism to achieve satisfactory scalability and is impractical in realistic scenarios. 
\end{remark}

\section{Advanced Approach: DP-starJ}

Motivated by the limitation in achieving an elegant tradeoff between utility, efficiency, and scalability under the output mechanism, we propose an advanced approach of DP-starJ, which can achieve strict \textsc{dp} with higher utility and efficiency to answer the star-join query. The main idea of DP-starJ is to decompose high sensitivity using the intrinsic characteristics of star-join to balance the utility and efficiency. Compared with the output perturbation mechanism, DP-starJ avoids the high sensitivity of star-join queries while improving the utility and reducing the computation cost. 

The overall intuition of DP-starJ is to add noise to star-join queries from the view of input, which turns out to be a challenging task. In the following, we first present a framework of DP-starJ to answer star-join query under \textsc{dp}
and then propose a mechanism of input perturbation in DP-starJ, namely, Predicate Mechanism (PM). Afterwards, we introduce DP-starJ to support for various types of star-join queries. At last, we give the granularity of privacy and utility study.

\subsection{DP-starJ} 

As discussed in Section~\ref{sec3}, none of the existing mechanisms can overcome all three key challenges (utility, efficiency, and scalability) in \textsc{dp}-compliant star-join query. To address this problem, we first propose a framework called DP-starJ that answers the star-join query under \textsc{dp}. Its main idea is to add random noise to star-join query procedure rather than the query result. DP-starJ decomposes the predicates of star-join query to reduce the high global sensitivity of the query. Specifically, DP-starJ mainly consists of three phases as depicted in Figure~\ref{fig_2}: 

\begin{figure*}[t]
  \centering
  \includegraphics[width=\linewidth]{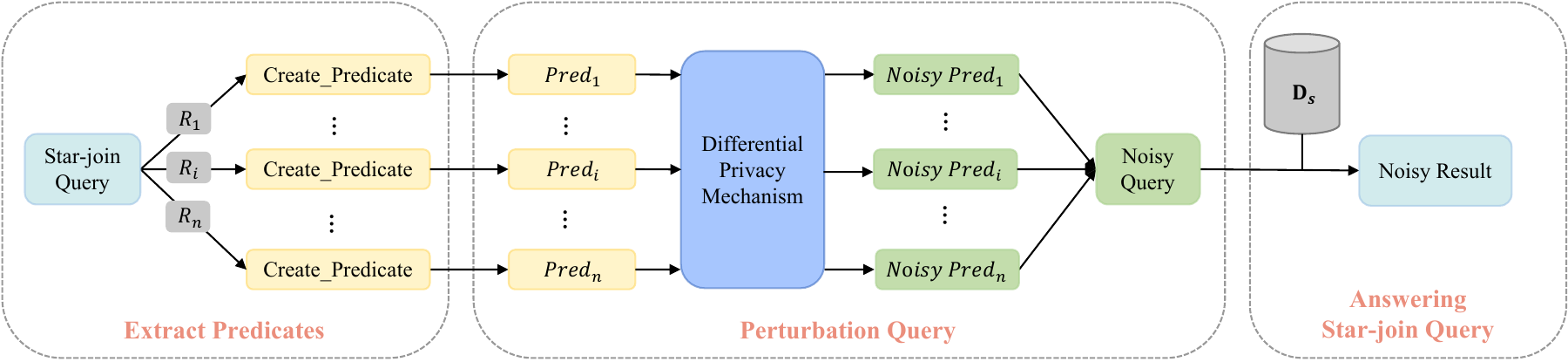}\vspace{-0.3cm}
  \caption{Execution Phases with DP-starJ}\label{fig_2}
\end{figure*}

\emph{Phase 1. Extract Predicates.} In DP-starJ, given the star-join query $Q$ with $n$ dimension tables and a fact table $R_0$, the server first extracts predicates from the query. The star-join query can be expressed as a predicate query because of the star structure of the database instance and the independence of dimension tables. Therefore, in this phase, the server mainly extracts predicates of each dimension table according to the star-join query and database schema. In star-join query, the type of predicate typically includes the range constraint and point constraint of the attributes in the dimension table. Hence, this phase extracts predicates based on the dimension table involved in the given query. If the star-join query includes all dimension tables, the server will create one predicate for each of the $n$ predicates for each of $n$ dimension tables, resulting in a total of $n$ predicates.

\emph{Phase 2. Perturbation Query.} In this phase, we employ some perturbation mechanisms to the star-join query to ensure differential privacy of the DP-starJ framework. The main process involves adding random noise into the predicates generated in Phase 1. Then, it aggregates all noise predicates together into noised star-join query, where the perturbation mechanism adopted is orthogonal and various specific methods can be employed, \eg the Laplace Mechanism for each attribute adopted in this paper. 


\emph{Phase 3. Answering Star-join Query.} In this phase, the server answers the star-join query $Q$ in a \textsc{dp} manner by accessing the database instance $\mathbf{D}_s$ with the noisy star-join query $\hat{Q}$. 

To balance the utility, efficiency, and scalability, DP-starJ responds to the star-join query in the form of an input perturbation. In addition, it decomposes predicates to reduce the sensitivity of the query in order to improve the utility. In the following, we present Predicate Mechanism to implement the DP-starJ framework, which helps us identify the key problems for developing DP-starJ.

\subsection{Predicate Mechanism}
Let $\mathbf{D}_s$ be a database instance over star schema and a star-join query $Q$ aggregates over the join result $J(\mathbf{D}_s)$. Since  $\Phi$ is an indicate function, we simplify  Equation~\ref{eq1} to the following form:
\begin{equation}
   Q(\mathbf{D}_s) = \sum_{t \in J(\mathbf{D}_s)} \Phi(t) \cdot \mathbf{w}(t) = \mathbf{\Phi} \cdot \mathbf{w}(t).
\end{equation}
Hereby $\mathbf{\Phi}$ refers to a predicate matrix of star-join query with all records. Since in star-join queries each dimension table $R_i, i \in [n]$ is independent of each other and places filter predicates towards the attributes locally, thus $\mathbf{\Phi}$ can reflect the conjunctions of the predicate, $\Phi := \phi_{a_1} \wedge \phi_{a_2} \wedge ...\wedge \phi_{a_n}$. Moreover, we can vectorize the weight function of tuple $\mathbf{w}(t)$ as $\mathbf{W}$, so the above equation can be transformed into the following form:
\begin{equation}
   Q(\mathbf{D}_s) =  \Phi \cdot \mathbf{w}(t) = \Phi \cdot \mathbf{W}= (\phi_{a_1} \wedge \phi_{a_2} \wedge ...\wedge \phi_{a_n}) \cdot \mathbf{W},
\end{equation}
where $\phi_{a_i}$ is the predicate condition of dimension table $R_i$ in the star-join query $Q$. 

\begin{algorithm}[t]
\caption{Predicate Mechanism}
\label{alg:pm}
\KwIn{Star-join query $Q$, Data instance $\mathbf{D}_s$, Data Matrix $\mathbf{W}$,  parameter $\epsilon$}
\KwOut{Noisy result: $\hat{Q}(\mathbf{D}_s)$}
$ \Phi \leftarrow  Q $ \;
$\epsilon_i = \frac{\epsilon}{n}$ \;
\For { each predicate $\phi_{a_i} \in \Phi$}{
$\hat{\phi}_{a_i} \leftarrow \phi_{a_i} + Lap(\frac{dom(a_i)}{\epsilon_i})$} 
$\hat{\Phi} \leftarrow \hat{\phi}_{a_1} \wedge...\wedge \hat{\phi}_{a_n}$ \;
$\hat{Q}(\mathbf{D}_s) = \hat{\Phi} \cdot \mathbf{W}$ \;
    $\mathsf{Return} \ \hat{Q}(\mathbf{D}_s)$
\end{algorithm}

Unlike the output perturbation, Predicate Mechanism adds random perturbations to the predicates involved in the star-join query procedure before touching the raw database instance. 
\begin{equation}
\begin{split}
   \hat{Q}(\mathbf{D}_s) &= (\Phi+Lap(\frac{GS_{\Phi}}{\epsilon})) \cdot \mathbf{W} \\
      &= ((\phi_{a_1}+ Lap(\frac{GS_{\phi_{a_1}}}{\epsilon_1})) \wedge...\wedge (\phi_{a_n}+ Lap(\frac{GS_{\phi_{a_n}}}{\epsilon_n}))) \cdot \mathbf{W} \\
      &= (\hat{\phi}_{a_1} \wedge...\wedge \hat{\phi}_{a_n}) \cdot \mathbf{W} \\
      &=  \hat{\Phi} \cdot \mathbf{W},
\end{split}
\end{equation}
where the privacy cost is $\epsilon_i = \frac{\epsilon}{n}$, and the global sensitivity $GS_{\phi_{a_i}}$ of each predicate $\phi_{a_i}$ is the domain size of attribute $a_i$ in dimension table $R_i$. Algorithm~\ref{alg:pm} shows the pseudo-code of PM, the server (i) generates the predicate $\Phi$ in a star-join query $Q$, (ii) decomposes $\Phi$ into dimension table predicates $\phi_{a_i}$ based on $Q$ and adds noise to the predicates $\phi_{a_i}$, and (iii) answers star-join query $Q$ according to the noised predicate $\hat{\Phi}$ and finally obtains the \textsf{dp} result $\hat{Q}(\mathbf{D}_s)$. The main idea of PM is to add random noise to each predicate $\phi_{a_i}$ of dimension tables in star-join query $Q$ because the predicates of each dimension table are independent of each other. We now carry on with the predicate perturbation of each single-dimension table for the predicate mechanism implementation.

In the predicate perturbation, a straightforward solution is to perturb each predicate separately using a single Laplace perturbation algorithm, such that every attribute is given a privacy budget $\epsilon_i = \epsilon/n$. Then, it is well known that the Laplace perturbation is suitable for real value, but the predicate of the query may contain point constraints and range constraints of an attribute. For two classes of predicates, we use Laplace noise to perturb predicates with point constraints and range constraints, respectively. The specific process is as follows.

\textbf{Predicate Perturbation for Each Single Attribute.} In an attribute $a_i$, the predicate $\phi_{a_i}$ in dimension table $R_i$ may contain either range constraints $a_i \in [l, r]$, or point constraints $a_i = v$. If the predicate is a point constraint, the predicate perturbation is directly adding the Laplace noise to the value $v$. When the predicate is a range constraints, $a_i \in [l, r]$, the predicate perturbation is to perturb both ends of the interval $[l, r]$ independently using a Laplace perturbation algorithm, such that every attribute is given a privacy budget $\epsilon/2$. The specific process is shown in Algorithm~\ref{alg:na}.   

\begin{algorithm}[t]
\caption{Predicate Mechanism for An Attribute ($\mathbf{PM_A}$)}
\label{alg:na}
\KwIn{Predicate $\phi_{a_i}$ of an attribute $a_i$, parameter $\epsilon$}
\KwOut{Noisy Predicate: $\hat{\phi}_{a_i}$}
\eIf{$\phi_{a_i}$ is $a_i = v$}
{
$\hat{v} = v + Lap(dom(a_i)/\epsilon)$ \;
$\hat{\phi}_{a_i} \leftarrow a_i = \hat{v}$
}
{
$\phi_{a_i} \leftarrow a_i \in [l, r]$ \;
$\hat{l} = l + Lap(\frac{2 \cdot dom(a_i)}{\epsilon})$ \;
$\hat{r} = r + Lap(\frac{2 \cdot dom(a_i)}{\epsilon})$ \;
\While{$\hat{l} < \hat{r}$}{
$\hat{\phi}_{a_i} \leftarrow a_i \in [\hat{l}, \hat{r}]$}
}
    $\mathsf{Return} \ \hat{\phi}_{a_i}$
\end{algorithm}

\subsection{DP-starJ Applications}
 To further boost the robust performance, in this section, we discuss specific solutions forthe  predicate mechanism in DP-starJ for various types of star-join tasks. The main idea of DP-starJ is to inject random data-independent noise into star-join query, which is an application of the predicate mechanism on different star-join queries. Therefore, we present the predicate mechanism for aggregated star-join queries, ``\textsf{Group\_By}'' operation, and star-join workload queries as follow.

\textbf{Predicate Mechanism for Aggregated Star-join Queries.} We now consider the case for the star-join aggregation query that aggregates the number of tuples that suit the filter conditions. In this case, the solution is to perturb each predicate independently using a single predicate perturbation algorithm (Algorithm~\ref{alg:na}), such that every attribute is given a privacy budget $\epsilon_i = \epsilon/n$. The specific process is shown in Algorithm~\ref{alg:cq}, where the data matrix is $\mathbf{1}$ in which the value of all tuples is 1. If the aggregation function is the \textsf{SUM} in the star-join, the element of the data matrix is the value of the attribute, which is the summation over the attributes in the star-join query. In addition, if the star-join query involves ``\textsf{Group\_By}'' operation, similar to \textsf{COUNT} queries and \textsf{SUM} queries, we shall only perturb the predicates of the query before ``\textsf{Group\_By}'' operation. Therefore, DP-StarJ supports not only ordinary aggregate queries but also ``\textsf{Group\_By}'' statement in star-joins.

\begin{algorithm}[t]
\caption{Predicate Mechanism for Star-join Counting Query}
\label{alg:cq}
\KwIn{Star-join counting query $Q_c$, Data instance $\mathbf{D}_s$, Data Matrix $\mathbf{W}$,  parameter $\epsilon$}
\KwOut{Noisy result: $\hat{Q}_c(\mathbf{D}_s)$}
$ \Phi \leftarrow  Q_c $ \;
$\epsilon_i = \frac{\epsilon}{n}$ \;
\For { each predicate $\phi_{a_i} \in \Phi$}
{
     $\hat{\phi}_{a_i} \leftarrow \mathbf{PM_{A}}(\phi_{a_i}, \epsilon_i)$
} 
$\hat{\Phi} \leftarrow \hat{\phi}_{a_1} \wedge...\wedge \hat{\phi}_{a_n}$ \;
$\hat{Q}_c(\mathbf{D}_s) = \hat{\Phi} \cdot \mathbf{W}$ \;
    $\mathsf{Return} \ \hat{Q}_c(\mathbf{D}_s)$
\end{algorithm}

\textbf{Predicate Mechanism for Star-join Workload Queries.} In addition, as workload tasks are ubiquitous in \textsc{olap} scenarios~\cite{rohm2000olap}, we extensively consider answering star-join workload queries under differential privacy by using PM. Given a workload of $l$ star-join queries $\mathbf{L}$, $\mathbf{L} = \{Q_1, Q_2,..., Q_l \}$. One straightforward solution is to process each query $Q_i$ independently by using the Predicate Mechanism. Unfortunately, this strategy fails to exploit the correlations between different queries, which has been exhaustively studied and justified to be valuable in designing a more effective \textsc{dp} solution~\cite{yuan2015optimizing,li2015matrix}. Consider a workload of three different queries, $Q_1$ is interested in the total number of products sold in the first half of this year, while $Q_2$ is interested in the total number of products sold in the second half of this year, and $Q_3$ asks for the total number throughout the whole year. Clearly, the three queries are correlated with each other as $Q_3 = Q_1 + Q_2$. Given that fact, an alternative strategy for answering these queries is to process only $Q_1$ and $Q_2$, and use their sum to answer $Q_3$. Inspired by this phenomenon, we propose a Workload Decomposition (WD) strategy to answer star-join workload queries under differential privacy in the following.

Consider the star-join workload queries $\mathbf{L} = \{Q_1, Q_2,..., Q_l \}$. According to our discussion in Section 3, each star-join query $Q_i$ can be represented by its predicate $\Phi_i$. Following that way, the star-join workload queries $\mathbf{L}$ can be accordingly represented as a set of predicates $\Phi_i$, $\mathbf{L} := \{\Phi_1, \Phi_2,..., \Phi_l \}$. Each predicate $\Phi_1$ refers to filter conditions for different dimension tables, $\Phi_i := \phi_{a_1}^{i} \wedge \phi_{a_2}^{i} \wedge ...\wedge \phi_{a_n}^{i}$. 

Firstly, we adopt one-hot-encoding to quantify $\Phi_i$ into a series of vectors. As shown in Example~\ref{ex1}, the predicate of the star-join query is $\Phi = \phi_{Date} \wedge \phi_{Cust} \wedge \phi_{Supp}$, assume that the domain of $region$ is $\{ \rm A, B, C\}$ and ${\rm REGION = C}$, we can vectorized $\Phi$ as $[111111000000001001]$. Similarly, the vector representation for $\phi_{Date}$ and $\phi_{Cust}$($\phi_{Supp}$) are $[111111000000]$ and $[001]$($[001]$), respectively. Therefore, the workload queries $\mathbf{L}$, \ie a collection of $l$ star-join queries, can be arranged by rows and forms an $l \times m_d$ matrix and $m_d = \prod^n_{i=1} m_i$. The predicate matrix $\mathbf{P}^{\mathbf{L}}_i$ of each dimension table $R_i$ on workload queries $\mathbf{L}$ is an $l \times m_i$ matrix, hereby $m_i$ is the domain size of attribute on dimension table $R_i$. 

Secondly, for each predicate matrix $\mathbf{P}^{\mathbf{L}}_i$, we shall perform a matrix decomposition as follow:
\begin{definition}[Matrix Decomposition]
Given a predicate matrix $\mathbf{M}$ and a strategy matrix $\mathbf{A}$, we say $\mathbf{M}$ decomposes into $\mathbf{XA}$ if each predicate in $\mathbf{M}$ can be expressed as a linear combination of predicates in $\mathbf{A}$. In other words, there exists a solution matrix $\mathbf{X}$ to $\mathbf{M} = \mathbf{XA}$.
\end{definition}
For each predicate matrix on the workload queries $\mathbf{L}$, $\mathbf{P}^{\mathbf{L}}_i$, MD shall finds a new strategy matrix $\mathbf{A}_i$ to support $\mathbf{P}^{\mathbf{L}}_i$, and then evaluates the strategy matrix $\mathbf{A}_i$ using the Predicate Mechanism to obtain a noisy strategy matrix $\mathbf{\hat{A}}_i$. Afterwards, we can reconstruct a noisy predicate matrix from the noisy strategy matrix $\mathbf{\hat{A}}_i$, $\mathbf{\hat{P}}^{\mathbf{L}}_i = \mathbf{A}_i^{+} \mathbf{\hat{A}}_i$. 

Finally, we connect the noisy predicate matrix $\mathbf{\hat{P}}^{\mathbf{L}}_i$ to each corresponding dimension table into the noisy predicate matrix $\mathbf{\hat{P}}$. The server answers the star-join workload queries $\mathbf{\hat{L}}$ in a DP manner by accessing the database instance with the noisy predicate matrix $\mathbf{\hat{P}}$ on the noisy star-join workload queries $\mathbf{\hat{L}}$.

Algorithm~\ref{alg:wq} outlines the above procedure. It first uses one-hot-encoding to represent the predicate of star-join the workload queries $\mathbf{L}$ and assigns the privacy budget to the predicate matrix that decomposes the predicate according to the dimension table (Lines 1-2). After that, matrix decomposition is performed on each predicate matrix $\mathbf{P}^{\mathbf{L}}_i$ to get the corresponding strategy matrix $\mathbf{A}_i$ and applies Predicate Mechanism on the strategy matrix to obtain a noised strategy matrix $\hat{\mathbf{A}_i}$, and reconstruct the noise-injected predicate matrix $\hat{\mathbf{P}}^{\mathbf{L}}_i$ through the noise-injected strategy matrix (Lines 3-9). Afterwards, it connects each noise-injected predicate matrix $\mathbf{\hat{P}}^{\mathbf{L}}_i$ into $\mathbf{\hat{P}}$ of the noisy star-join workload queries $\mathbf{\hat{L}}$ and obtain the noisy result $\hat{Q}_L(\mathbf{D}_s)$ of the workload by accessing the database instance $\mathbf{D}_s$ (Lines 10-12), and finally outputs the noisy result $\hat{Q}_L(\mathbf{D}_s)$.

\textbf{Predicate Mechanism for snowflake queries.} Besides star-join, the proposed PM can also be applied to the snowflake model (\resp snowflake query), which further hierarchizing the dimension tables of the star schema, resulting in a more normalized structure. For example, in Figure~\ref{fig_0} and Example~\ref{ex1}, $Date$ can be decomposed into dimension tables such as Year, Quarter, Month, and Day, reducing redundancy. Therefore, the star-join query in Example~\ref{ex1} can be extended to snowflake query by changing ${Date}.{month} < 7$ to ${Date}.MK = {Month}.MK$ $AND$ ${Month}.month < 7$. At this point, we can directly apply PM to perturb the predicate to obtain the \textsc{dp} snowflake query. This does not affect the functionality of DP-starJ while extending star-join queries to queries on the snowflake model.
\begin{algorithm}[t]
\caption{Predicate Mechanism for Star-join Workload Queries}
\label{alg:wq}
\KwIn{Star-join Workload Queries $\mathbf{L} = \{Q_1, Q_2,..., Q_l \}$, Data instance $\mathbf{D}_s$, Data Matrix $\mathbf{W}$,  parameter $\epsilon$}
\KwOut{Noisy result: $\hat{Q}_L(\mathbf{D}_s)$}
$ \mathbf{P} \leftarrow \mathbf{One-Hot-Encoding}(\mathbf{L}) $ \;
$\epsilon_i = \frac{\epsilon}{n}$ \;
\For { each predicate matrix $\mathbf{P}^{\mathbf{L}}_i \in \mathbf{P}$}
{
     $\mathbf{A}_i \leftarrow\mathbf{MatrixDecom}(\mathbf{P}^{\mathbf{L}}_i)$ \;
     $\phi_{a_i} \leftarrow \mathbf{A}_i$ \;
     $\hat{\phi}_{a_i} \leftarrow \mathbf{PM_A}(\phi_{a_i}, \epsilon_i) $ \;
     $\hat{\mathbf{A}}_i \leftarrow \hat{\phi}_{a_i}$ \;
     $\hat{\mathbf{P}}^{\mathbf{L}}_i = \mathbf{A}^{+}_i\hat{\mathbf{A}_i}$ \;
} 
$\hat{\mathbf{P}} \leftarrow {\hat{\mathbf{P}}^{\mathbf{L}}_1,...,\hat{\mathbf{P}}^{\mathbf{L}}_n}$ \;
$\hat{\mathbf{L}} \leftarrow \hat{\mathbf{P}}$ \;
$\hat{Q}_L(\mathbf{D}_s) = \hat{\mathbf{L}} \cdot \mathbf{W}$ \;
    $\mathsf{Return} \ \hat{Q}_L(\mathbf{D}_s)$
\end{algorithm} 

\subsection{Theoretical Study over the Privacy and Utility}\label{ssec54}
We now conduct a theoretical study on the privacy guarantee as well as the utility of the proposed Predicate Mechanism. In this section, we first study the privacy guarantee of the Predicate Mechanism and DP-starJ in terms of Definition~\ref{df:36}. After that, we theoretically study the utility of the Predicate Mechanism.

\vspace{-2ex}\begin{theorem}\label{theo:51}
  Algorithm~\ref{alg:na} satisfies $\epsilon$-differential privacy.
\end{theorem}
\vspace{-2ex}\begin{proof}
    Algorithm~\ref{alg:na} in the paper adds Laplace noise to the predicate, and the scale of Laplace is the ratio of the domain and privacy cost. In the worst case, the number of ways a change in a record can affect the predicate is equal to the size of the domain. Therefore, its global sensitivity is the size of the domain. In other words, Algorithm~\ref{alg:na} essentially implements the Laplace mechanism on the predicate. According to the Theorem~\ref{theo:41}, Algorithm~\ref{alg:na} satisfies $\epsilon$-differential privacy.
\end{proof}

\vspace{-2ex}\begin{theorem}
  Predicate Mechanism and DP-starJ satisfy $\epsilon$-differential privacy.
\end{theorem}
\vspace{-2ex}\begin{proof}
The Predicate Mechanism decomposes $\Phi$ into dimension table predicates $\phi_{a_i}$ based on $Q$ and adds noise to the predicates $\phi_{a_i}$, the proof of PM is transformed into proving that each predicates $\hat{\phi}_{a_i}$ satisfies $\frac{\epsilon}{n}$-differential privacy according to the Theorem~\ref{theo:51}, and whether the predicate $\hat{\Phi}$ satisfies $\epsilon$-differential privacy. According to the fact that $\Phi := \phi_{a_1} \wedge \cdot \cdot \cdot\wedge \phi_{a_1}$ and each $\phi_{a_1}$ is independent of each other, we thus have $Pr[\Phi] = Pr[\phi_{a_1} \wedge \cdot \cdot \cdot\wedge \phi_{a_1}] = Pr[\phi_{a_1}] \cdot Pr[\phi_{a_2}]... Pr[\phi_{a_n}]$. \\ Meanwhile, each predicates $\phi_{a_i}$ satisfies $\frac{\epsilon}{n}$-differential privacy, $Pr[\phi_{a_i}] \leq e{^{\frac{\epsilon}{n}}}\cdot Pr[\hat{\phi}_{a_i}]$. Therefore, 
\begin{equation}
\begin{split}
    Pr[\Phi] &= Pr[\phi_{a_1}] \cdot  Pr[\phi_{a_2}] \cdot \cdot \cdot Pr[\phi_{a_n}] \\
    &\leq e{^\epsilon} \cdot (Pr[\hat{\phi}_{a_1}]\cdot Pr[\hat{\phi}_{a_2}]\cdot \cdot \cdot Pr[\hat{\phi}_{a_n}]) = e{^\epsilon} \cdot Pr[\hat{\Phi}]
\end{split}
\end{equation}
The Predicate Mechanism satisfies $\epsilon$-differential privacy. Similar to Predicate Mechanism, DP-starJ can be proved to satisfy $\epsilon$-differential privacy in the same way.
\end{proof}

Both Algorithm~\ref{alg:cq}\&\ref{alg:wq} adopt the Predicate Mechanism, we shall study the privacy guarantee of them accordingly as follows.
\vspace{-2ex}\begin{theorem}
  Algorithm~\ref{alg:cq} satisfies $\epsilon$-differential privacy.
\end{theorem}
\vspace{-2ex}\begin{proof}
 Algorithm~\ref{alg:cq} decomposes query predicates by dimension tables, allocating the privacy budget of $\frac{\epsilon}{n}$ to each predicate. According to Theorem~\ref{theo:51}, each noisy predicate satisfies $\frac{\epsilon}{n}$-differential privacy via Algorithm~\ref{alg:na}. Within each predicate, there is sequential composition because adding or removing a record affects all predicates. According to Sequential Composition~\cite{dwork2014algorithmic}, Algorithm~\ref{alg:cq} satisfies $\epsilon$-differential privacy. 
 \end{proof}
\begin{table*}[t]
  \centering
  \tabcolsep = 0.026\linewidth
    \caption{Relative error(\%) of various mechanisms PM, R2T, LS on SSB queries by varying $\epsilon$.}\label{tab:t1}
  \begin{tabular}{c|c||c|c|c|c||c|c|c||c|c}
    \hline
    \multicolumn{2}{c||}{Query type} & \multicolumn{4}{c||}{\textsf{COUNT}} & \multicolumn{3}{c||}{\textsf{SUM}} & \multicolumn{2}{c}{\textsf{GROUP BY}}  \\
    \hline
    \multicolumn{2}{c||}{Query}&$Q_{c1}$&{$Q_{c2}$}&{$Q_{c3}$}&{$Q_{c4}$}&{$Q_{s2}$}&{$Q_{s3}$}&{$Q_{s4}$}& {$Q_{g2}$}&{$Q_{g4}$} \\
    \hline\hline
    \multirow{3}*{$\epsilon = 0.1$}&PM&11.89 &9.46&19.02&8.22&12.07&16.3&17.36&11&28.63 \\
    \cline{2-11}
    &R2T& 120.87 &41.51&29.63&20.41&80.61&80.22&80.14&\multicolumn{2}{c}{Not supported*} \\ 
    \cline{2-11}
     & LS& 180.9 &73.39&78.12&80.44&\multicolumn{5}{c}{Not supported} \\ 
    \hline\hline
    \multirow{3}*{$\epsilon = 0.2$} & PM & 11.93 &9.28&16.48&5.12&11.55&13.07&12.39&10.6&18.8 \\
    \cline{2-11}
    & R2T & 59.76 &30.38&19.4&15.16&79.91&80.17&79.83&\multicolumn{2}{c}{Not supported*} \\ 
    \cline{2-11}
     & LS& 121.68 &61.8&58.61&83&\multicolumn{5}{c}{Not supported} \\ 
    \hline\hline
     \multirow{3}*{$\epsilon = 0.5$} & PM&8.66&7.61&15.42&4.3&11.58&12.45&10.43&9.88&11.83 \\
    \cline{2-11}
    & R2T & 84.48 &22.9&19.67&11.55&79.46&80.08&79.61&\multicolumn{2}{c}{Not supported*} \\ 
    \cline{2-11}
     & LS& 86.84 &47.6&20.38&52.09&\multicolumn{5}{c}{Not supported} \\
    \hline\hline
     \multirow{3}*{$\epsilon = 0.8$}& PM &5.1&7.86&13.35&3.71&11.43&12.59&7.58&9.25&6.45 \\
    \cline{2-11}
    & R2T & 76.16&17.46&14.56&9.21&79.21&80.03&79.17&\multicolumn{2}{c}{Not supported*} \\ 
    \cline{2-11}
     & LS &77.23&32.99&13.28&31.89&\multicolumn{5}{c}{Not supported} \\ 
    \hline\hline
     \multirow{3}*{$\epsilon = 1$} & PM&5&7.53&11.76&1.92&10.51&12.18&5.02&8.99&4.02 \\
    \cline{2-11}
    & R2T&61.77&13.1&15.63&7.71&79.04&79.97&79.38&\multicolumn{2}{c}{Not supported*} \\ 
    \cline{2-11}
     & LS &84.06&27.99&20.19&14.97&\multicolumn{5}{c}{Not supported} \\ 
    \hline
  \end{tabular}
  \begin{tablenotes}
  \footnotesize
   \item{*} It is a future work of ~\cite{dong2022r2t}. 
  \end{tablenotes}
\end{table*}
\begin{theorem}
  Algorithm~\ref{alg:wq} satisfies $\epsilon$-differential privacy.
\end{theorem}

\begin{proof}
Algorithm~\ref{alg:wq} is suitable to answer star-join workload queries indirectly, by first perturbing a set of intermediate predicates under differential privacy via Algorithm~\ref{alg:na}, and then combining their predicates to answer the star-join workload queries. Thus, Algorithm~\ref{alg:wq} satisfies $\epsilon$-differential privacy.
\end{proof}

Besides the privacy guarantee, we now conduct a theoretical study on the utility of the Predicate Mechanism.  

\begin{theorem}[Loose Bound of Predicate Mechanism]
  Let $Q$ be the star-join query with $n$ dimension tables and a fact table, the variance of using Predicate Mechanism is $(\frac{2n^2}{\epsilon^2})^n \cdot \prod_{i=1}^n dom(a_i)^2$. 
\end{theorem}
\begin{proof}
  Since the Predicate of star-join query, \ie $\Phi$, is the conjunction of each predicate of dimension tables $\phi_{a_i}$, $\Phi = \prod_{i=1}^n \phi_{a_i}$, and each predicate $\hat{\phi}_{a_i}$ satisfies $\frac{\epsilon}{n}$-differential privacy, the variance of each predicates $\hat{\phi}_{a_i}$ is $2(\frac{n \cdot dom(a_i)}{\epsilon})^2$ and the expectation is 0 due to the Laplace noise. In addition, as each dimension table is independent of each other, the variance of predicate mechanism is the multiplication of the variance of $\hat{\phi}_{a_i}$, $(\frac{2n^2}{\epsilon^2})^n \cdot \prod_{i=1}^n dom(a_i)^2$.
\end{proof}

\begin{theorem}[Tight Bound of Predicate Mechanism]
  Let $Q$ be the star-join query with $n$ dimension tables and a fact table, the variance of using the Predicate Mechanism is $(\frac{2n^2}{\epsilon^2}) \cdot \sum_{i=1}^n dom(a_i)^2$. 
\end{theorem}
\begin{proof}
  Since the input of $\Phi$ is the conjunction of the binary, then $\Phi$ can be expressed as an indicate function, $\Phi = \mathbb{I}[\sum_{i=1}^n \phi_{a_i} = n]$. Moreover, each predicate $\hat{\phi}_{a_i}$ satisfies $\frac{\epsilon}{n}$-differential privacy, the variance of each predicates $\hat{\phi}_{a_i}$ is $2(\frac{n \cdot dom(a_i)}{\epsilon})^2$ and the expectation is 0 due to the Laplace noise. In addition, the dimension table is independent of each other and the introduction of an indication function does not cause any extra errors. Therefore, the variance of predicate mechanism is the sum of the variance of $\hat{\phi}_{a_i}$, $(\frac{2n^2}{\epsilon^2}) \cdot \sum_{i=1}^n dom(a_i)^2$.
\end{proof}

\begin{table*}[h]
  \centering
  \tabcolsep = 0.015\linewidth
    \caption{Comparison between PM, R2T, TM on $k$-star queries by varying $\epsilon$.}\label{tab:t2}
  \begin{tabular}{c|c|c||c|c||c|c||c|c}
    \hline
    \multicolumn{3}{c||}{Privacy budget} & \multicolumn{2}{c||}{$\epsilon = 0.1$} & \multicolumn{2}{c||}{$\epsilon = 0.5$} & \multicolumn{2}{c}{$\epsilon = 1$}  \\
    \hline
    \multicolumn{3}{c||}{Result type}&Relative error(\%)&Time(s)&Relative error(\%)&Time(s)&Relative error(\%)&Time(s) \\
    \hline\hline
    \multirow{6}*{\textbf{Deezer}}&\multirow{3}*{$Q_{2\ast}$}&PM&38.25&0.14&35.91&0.11&30.53&0.11 \\
    \cline{3-9}
    &&R2T&52.45&15.02&74.56&15.03&63.36&15.46 \\ 
    \cline{3-9}
     &&TM&2431.4&5.53&339.55&5.27&279.18&4.9 \\ 
    \cline{2-9}
    &\multirow{3}*{$Q_{3\ast}$} &PM&65.06&0.84&58.85&1.25&56.67&1.15 \\
    \cline{3-9}
    &&R2T&\multicolumn{6}{c}{Over time limit} \\ 
    \cline{3-9}
     &&TM&385.75&164.05&306.49&164.45&117.3&160.77 \\ 
    \hline\hline
    \multirow{6}*{\textbf{Amazon}}&\multirow{3}*{$Q_{2\ast}$}&PM&17.67&0.67&11.41&0.60&7.39&0.75 \\
    \cline{3-9}
    &&R2T&23.91&127.25&10.63&131.86&8.38&145.39 \\ 
    \cline{3-9}
     &&TM&3750.34&80.4&482.01&83.51&42.03&76.33 \\ 
    \cline{2-9}
    &\multirow{3}*{$Q_{3\ast}$} &PM&16.25&4.62&14.78&4.70&7.90&4.33 \\
    \cline{3-9}
    &&R2T&\multicolumn{6}{c}{Over time limit} \\ 
    \cline{3-9}
     &&TM&\multicolumn{6}{c}{Over time limit} \\
    \hline
  \end{tabular}

\end{table*}

\begin{figure}[h]
  \centering
  \includegraphics[width=\linewidth]{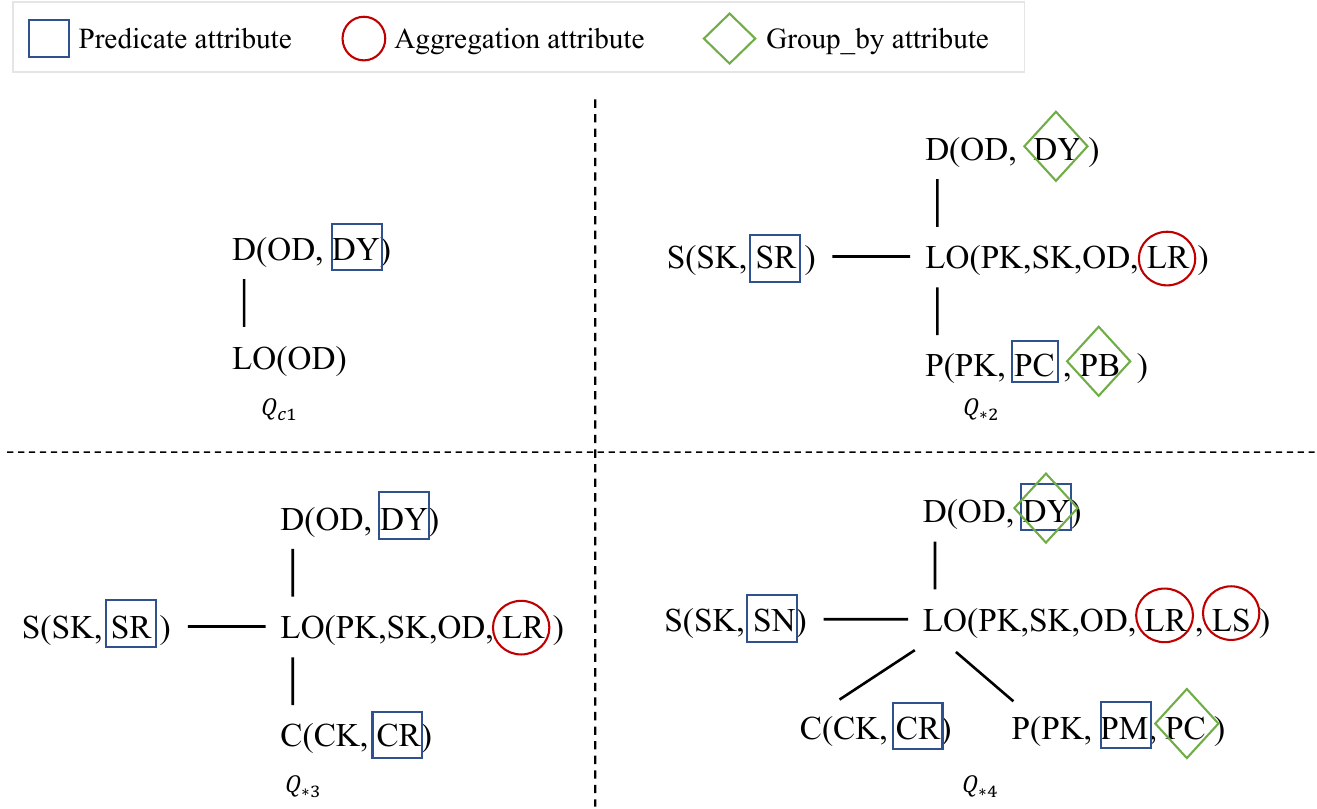}\vspace{-0.3cm}
  \caption{The structure of SSB queries}\label{fig_13}
\end{figure}

\begin{figure*}[h]
  \centering
  \includegraphics[width=\linewidth,height=7.5cm]{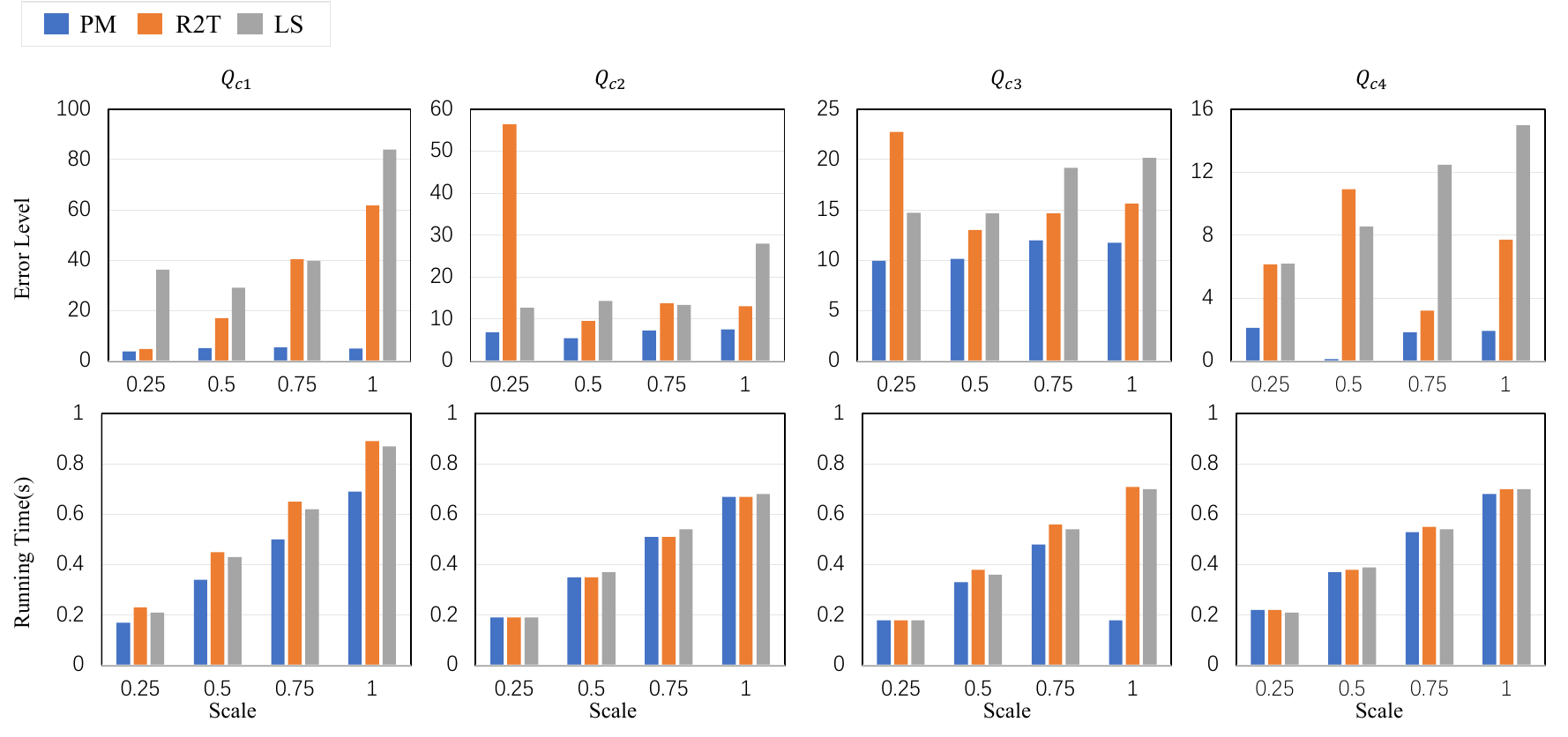}\vspace{-0.3cm}
  \caption{Running times and error level of PM, R2T, LS for different data scales (\textsf{COUNT}).}\label{fig_3}
\end{figure*}

\begin{figure*}[h]
  \centering
  \includegraphics[width=0.8\linewidth,height=7.5cm]{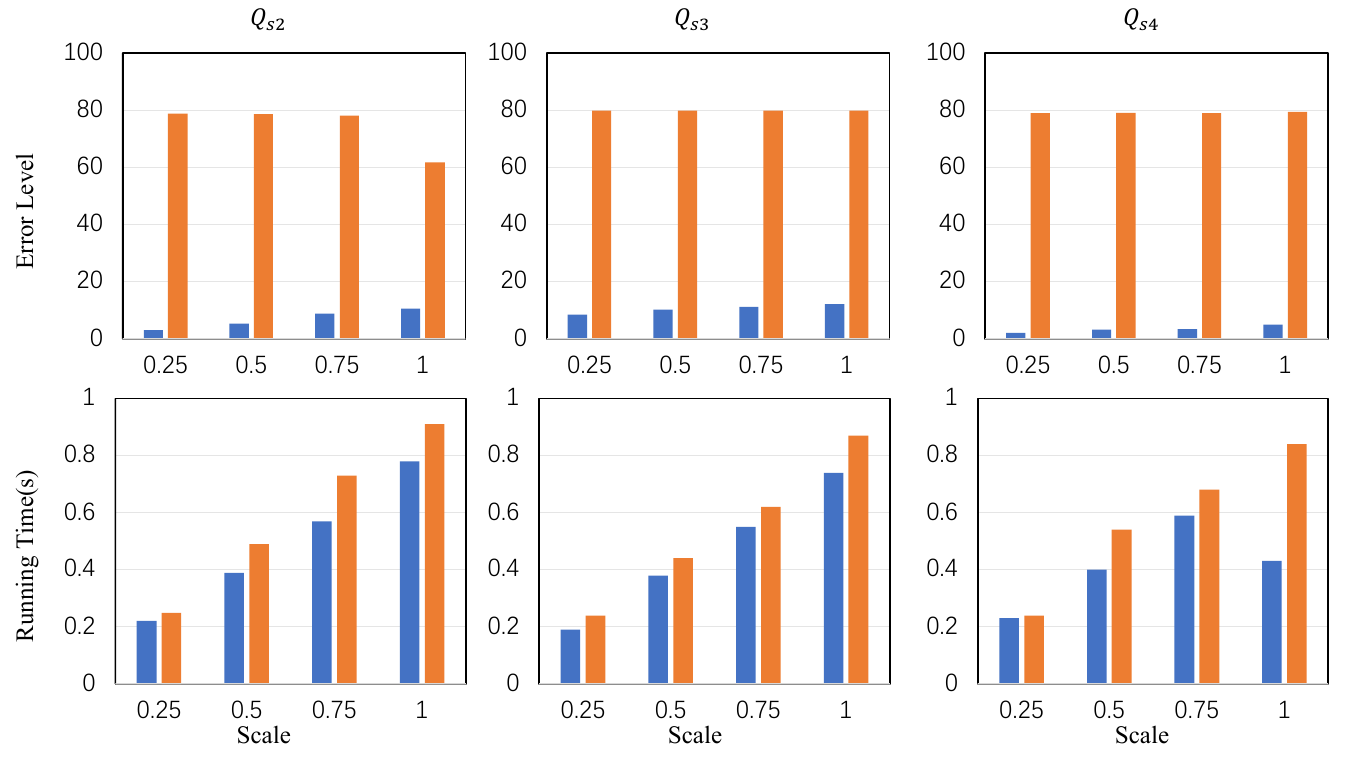}\vspace{-0.3cm}
  \caption{Running times and error level of PM and R2T for different data scales (\textsf{SUM}).}\label{fig_6}
\end{figure*}

\begin{figure*}[t]
  \centering
  \includegraphics[width=\linewidth]{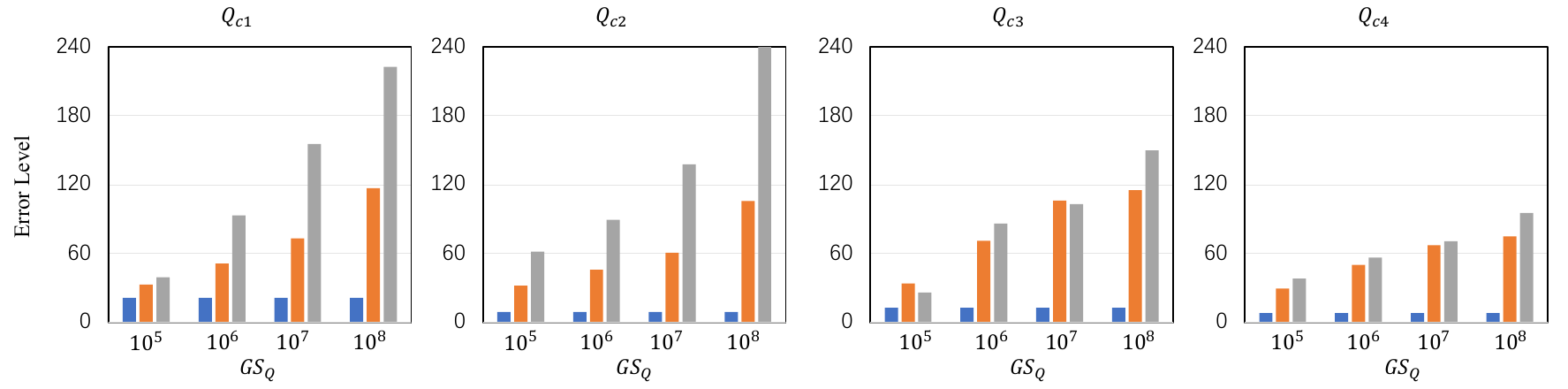}\vspace{-2ex}
  \caption{Error level of PM, R2T, LS for different $GS_Q$.}\label{fig_4}\vspace{-1ex}
\end{figure*}

\begin{figure*}[t]
  \centering
  \includegraphics[width=0.8\linewidth,height=7.5cm]{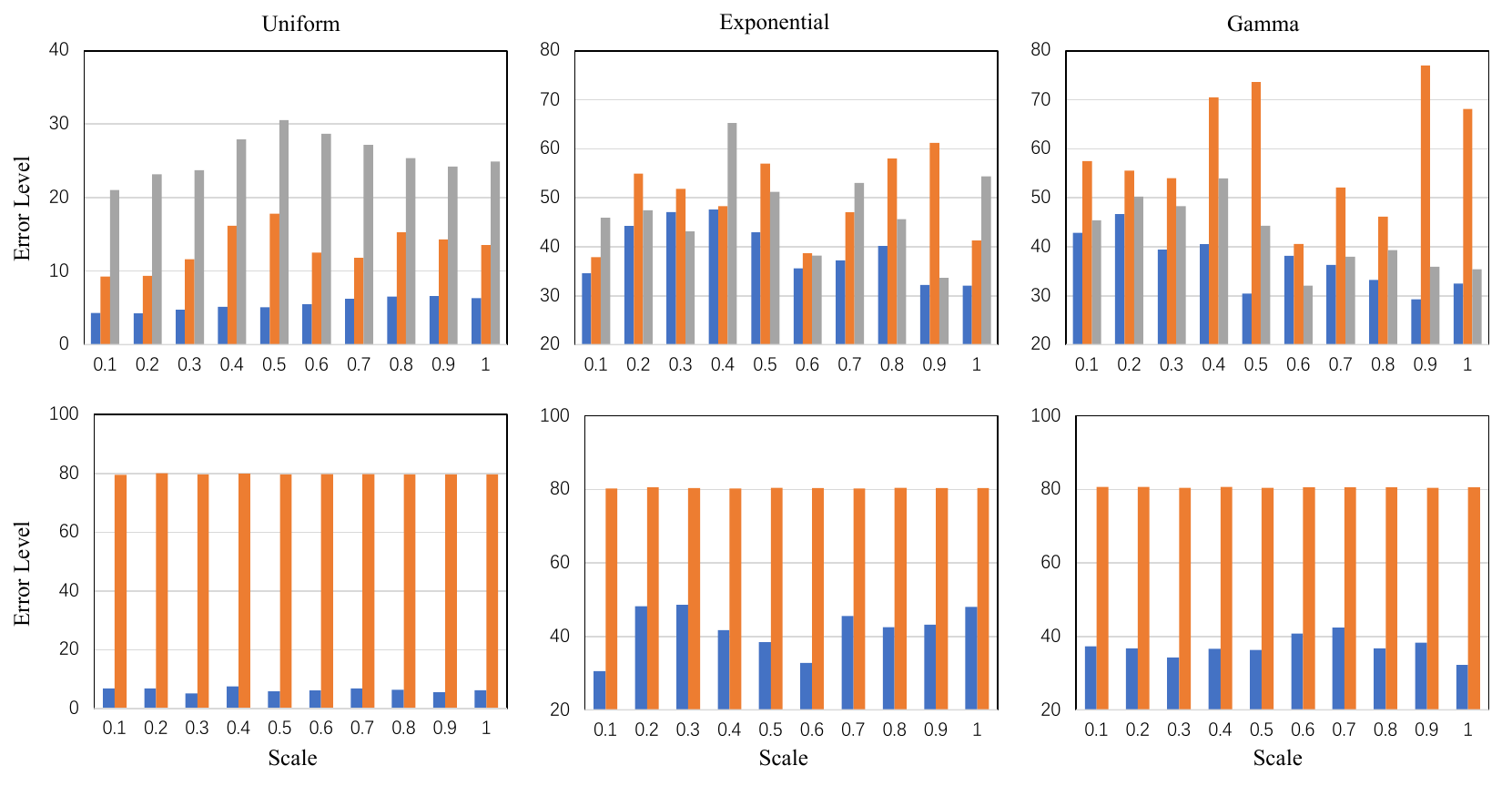}\vspace{-2ex}
  \caption{Error level of PM, R2T, LS for different distributions on $Q_{c3}$(top) and $Q_{s3}$(bottom) with different data scales.}\label{fig_9}\vspace{-1ex}
\end{figure*}

\begin{figure}[t]
  \centering
  \includegraphics[width=0.8\linewidth]{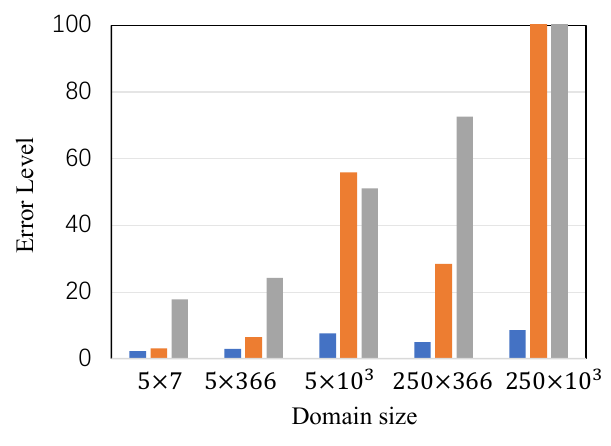}\vspace{-2ex}
  \caption{Error level of PM, R2T, LS for different domain sizes.}\label{fig_10}\vspace{-2ex}
\end{figure}

\section{Experiments}
We conducted empirical studies to test the performance of our model on a pair of benchmarking datasets. To evaluate the performance in various applications, the empirical study involves not only general star-join queries in \textsc{olap} scenarios, but also $k$-star counting queries in the graph, which is a fundamental task in graph database and representative instance of star-join in specific applications.  

For general star-join queries, we compare with a pair of state-of-the-art \textsc{dp}-compliant query schemes, namely R2T~\cite{dong2022r2t} and the local sensitivity-based mechanism (LS) \cite{tao2020computing}. 
For $k$-star counting queries, in line with~\cite{dong2022r2t}, we compare PM with R2T and naive truncation with smooth sensitivity (TM) \cite{kasiviswanathan2013analyzing}, which dominates LS in \textsc{dp}-compliant $k$-star tasks. 

\noindent\emph{\textbf{Dataset.}} 
To test the performance on general star-join queries, we perform experiments using the Star Schema Benchmark (SSB)~\cite{o2007star}, a variation of the TPC-H benchmark widely adopted in star-join studies~\cite{o2009star, sanchez2016review}. It changes the snowflake model adopted in TPC-H into a star model. SSB has a fact table and four dimension ones. Each dimension table contains hierarchical attributes, the value of which can be categorized into three types based on the hierarchy: large, medium, and small. For example, the $Customer$ table contains attributes with different domain values of \emph{city}, \emph{region}, and \emph{address}. In a star join query, the predicate only involves one of them.


For $k$-star counting queries, we adopt two real-world network datasets~\cite{leskovec2016snap}, namely \textbf{Deezer} and \textbf{Amazon}.The former collects all friendship relations of users from 3 European countries using the music streaming service Deezer, containing 144,000 nodes (\ie users) and 847,000 edges (\ie friendships). The latter is an Amazon co-purchasing network, which contains 335,000 nodes and 926,000 edges. The $k$-star counting queries predicate refers to its node range, so the domain size is its number of vertices.

\subsection{Setup}
\emph{\textbf{Queries.}} 
We test 9 queries out of three standard star-join tasks from SSB, including counting queries $\{Q_{c1}, Q_{c2}, Q_{c3}, Q_{c4}\}$, sum queries $\{Q_{s2}, Q_{s3}, Q_{s4}\}$, and group-by queries $\{Q_{g2}, Q_{g4}\}$. As an example, $\{Q_{c1}\}$ involves a dimension table, $\{Q_{c2}, Q_{c3}\}$ contains 3 dimension tables, and $Q_{c4}$ involves all the dimension tables. The structure of these queries are outlined in Figure~\ref{fig_13}.

For star-join workload queries, we utilize two types of the counting queries, $\{W_{1}, W_{2}\}$. $W_{1}$ contains all point constraints for one of three dimension tables. $W_{2}$ contains constraints for three dimension tables, one of which is a cumulative distribution (\ie each query sums the unit counts in a range $[1, i]$, where $i$ is in the domain of an attribute). The $\{W_{1}, W_{2}\}$ are as follows:
\vspace{-1ex}
\begin{equation}
    W_1 = \left [
    \setlength{\arraycolsep}{0.6pt}
    \begin{array}{ccccccc:ccccc:ccccc}
    1&0&0&0&0&0&0&1&0&0&0&0&1&0&0&0&0 \\
    0&1&0&0&0&0&0&1&0&0&0&0&1&0&0&0&0 \\
    0&0&1&0&0&0&0&1&0&0&0&0&1&0&0&0&0 \\
    0&0&0&1&0&0&0&1&0&0&0&0&1&0&0&0&0 \\
    0&0&0&0&1&0&0&1&0&0&0&0&1&0&0&0&0 \\
    0&0&0&0&0&1&0&1&0&0&0&0&1&0&0&0&0 \\
    0&0&0&0&0&0&1&1&0&0&0&0&0&1&0&0&0 \\
    0&0&1&1&0&0&0&0&1&0&0&0&0&1&0&0&0 \\
    0&0&0&1&1&0&0&0&0&1&0&0&0&1&0&0&0 \\
    0&0&0&0&1&1&0&0&0&0&1&0&0&1&0&0&0 \\
    0&0&0&0&0&1&1&0&0&0&0&1&0&1&0&0&0 
    \end{array}
    \right ],
    W_2 = \left [
    \setlength{\arraycolsep}{0.6pt}
    \begin{array}{ccccccc:ccccc:ccccc}
    1&0&0&0&0&0&0&0&0&1&0&0&1&0&0&0&0 \\
    1&1&0&0&0&0&0&0&0&1&0&0&1&0&0&0&0 \\
    1&1&1&0&0&0&0&1&0&0&0&0&1&0&0&0&0 \\
    1&1&1&1&0&0&0&0&0&1&0&0&0&1&0&0&0 \\
    1&1&1&1&1&0&0&0&0&0&1&0&0&0&1&0&0 \\
    1&1&1&1&1&1&0&0&0&0&0&1&1&0&0&0&0 \\
    1&1&1&1&1&1&1&0&0&1&0&0&0&1&0&0&0 
    \end{array}
    \right ].
    \nonumber
\end{equation}

For $k$-star counting queries, we test two different tasks: 2-star counting $Q_{2\ast}$ and 3-star counting $Q_{3\ast}$.

\noindent\emph{\textbf{Evaluation Metrics.}} Relative error is used as the utility measure and the privacy budget is varied from \{0.1, 0.2, 0.5, 0.8, 1\}. In addition, we also evaluate the running time for all the compared solutions. 

\subsection{Empirical results}
In each experiment, we report the average response time of 10 independent runs, each of which is kept within a time limit (\ie 3 hours).

\noindent\emph{\textbf{Utility.}} We test the utility of different solutions by varying $\epsilon$, and the results are shown in Table~\ref{tab:t1} and~\ref{tab:t2}, respectively. As the privacy budget increases, the error level gradually decreases as expected. In particular, according to Table~\ref{tab:t1}, both PM and R2T achieve high utility under star-join count queries, while LS achieves poor utility except for very large $\epsilon$. R2T achieves similar utility as PM on counting queries, but is much worse on sum queries. Table~\ref{tab:t1} shows that PM achieves order-of-magnitude improvements over R2T and LS in terms of utility. More importantly, PM supports a wider variety of star-join queries than R2T and LS. Remarkably, in all star-join queries, PM consistently achieves errors below 20\% (even $<15\%$ when $\epsilon\ge 0.5$). . 

Obviously, PM performs better on the SSB dataset. This is because the error of PM is proportional to the sum of domains according to our theoretical study in Section~\ref{ssec54}. Therefore, larger dimension tables in star-join queries lead to smaller relative errors. Compared with R2T and LS, PM exhibits much little change by varying $\epsilon$. In general, the DP-starJ is more stable and accurate than R2T and LS in light of star-join queries. Similarly, Table~\ref{tab:t2} also justifies the superiority of PM in terms of utility on \textbf{Deezer} and \textbf{Amazon} for $k$-star queries, offering order-of-magnitude improvements over other methods in different cases. In workload queries, the error level of PM and WD mechanisms are shown in Figure~\ref{fig_5}. As the figure demonstrates, WD always introduces lower error than PM, especially on $W_{1}$.
\begin{figure}[t]
  \centering
  \includegraphics[width=\linewidth]{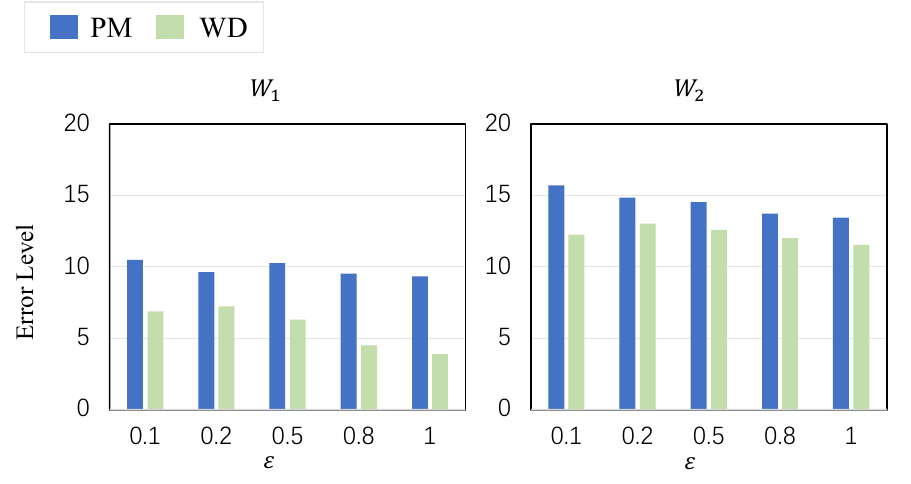}\vspace{-2ex}
  \caption{Error level of PM and WD for different $\epsilon$.}\label{fig_5}\vspace{-2ex}
\end{figure}

\begin{figure}[t]
  \centering
  \includegraphics[width=\linewidth]{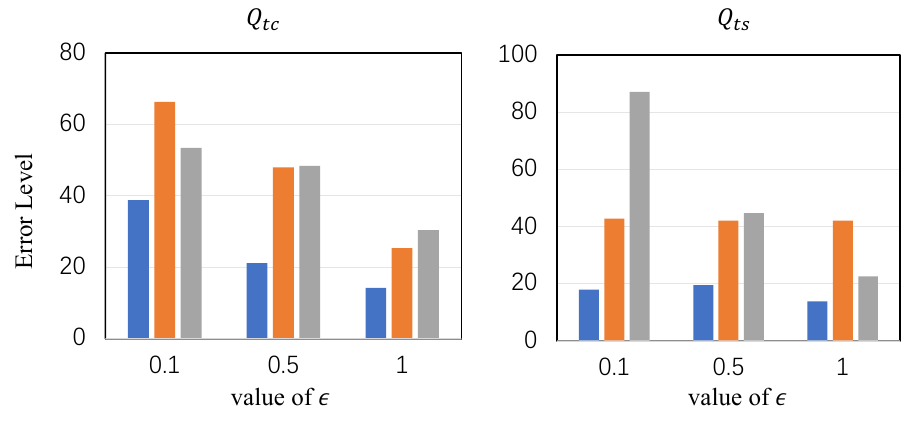}\vspace{-2ex}
  \caption{Error levels of various mechanisms on TPC-H queries by varying $\epsilon$.}\label{fig_11}\vspace{-2ex}
\end{figure}

\begin{figure*}[t]
  \centering
  \includegraphics[width=0.8\linewidth,height=7.5cm]{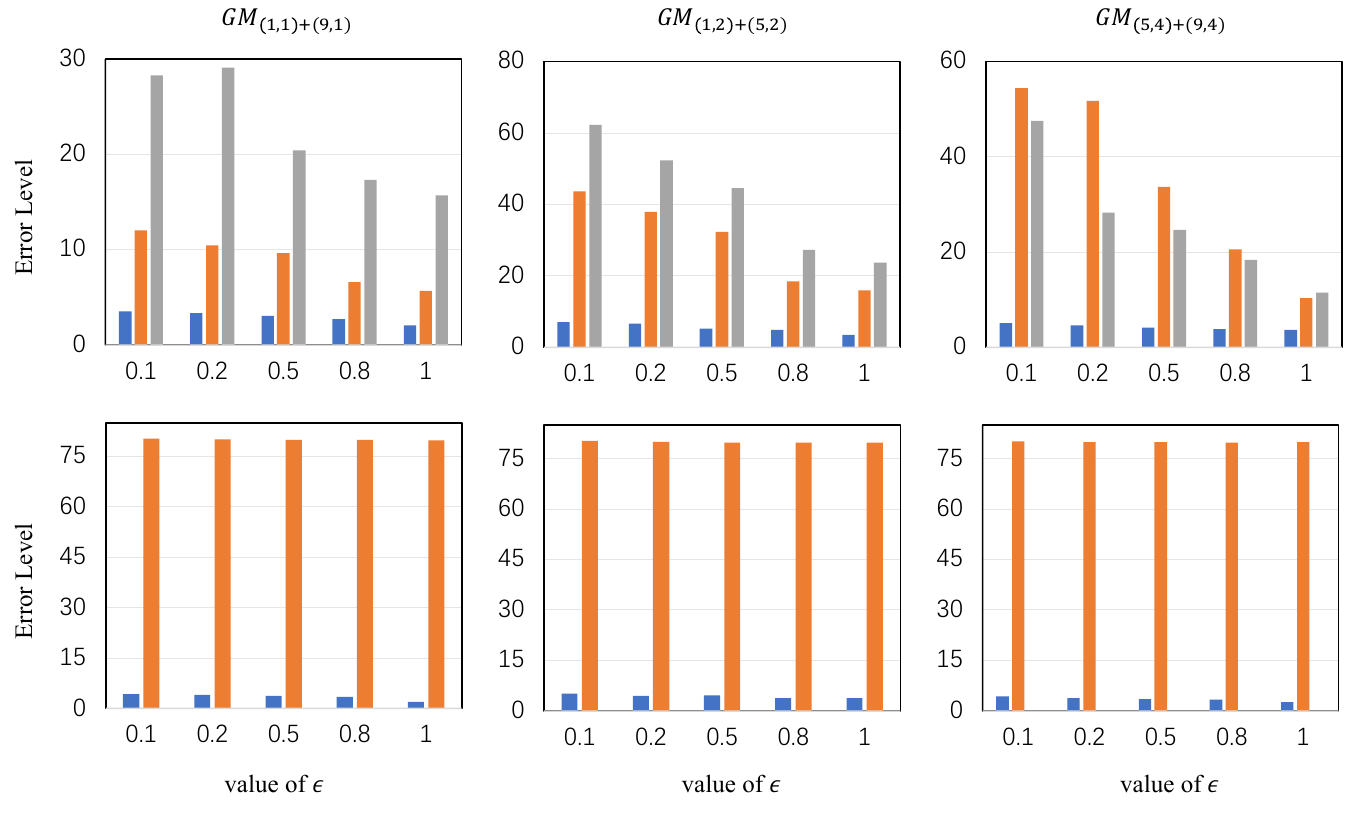}\vspace{-2ex}
  \caption{Error level of PM, R2T, LS for mixture of Gaussian distributions on $Q_{c3}$(top) and $Q_{s3}$(bottom) with different skewed parameter.}\label{fig_12}\vspace{-1ex}
\end{figure*}

\noindent\emph{\textbf{Efficiency.}} The running time of all mechanisms over the $k$-star counting queries are shown in Table~\ref{tab:t2}. On \textbf{Deezer} dataset, across all the compared solutions, R2T can only complete within the 3-hour time limit on $3$-star queries, although it achieves smaller errors on $2$-star queries than TM. Compared with R2T and TM, PM is much faster (at least 40 times faster than TM) as it does not require additional data truncation. Both R2T and TM exceed the time limit on \textbf{Amazon} dataset, which can be attributed to their increased number of joins on this larger-scale dataset. Additionally, R2T needs to solve linear programming problems to determine truncation thresholds and TM involves local sensitivity computation, both of which lead to extra computational overhead.  On the other hand, as the running time on SSB does not vary much either across approaches or under different privacy budgets, as have been shown in Figure~\ref{fig_3} and ~\ref{fig_6}, we choose not to explicitly report it in Table~\ref{tab:t1}.

\noindent\emph{\textbf{Scalability.}} In addition, we also test the scalability of the approaches by varying the volume of the database, using SSB with scale factors ranging from 0.25 to 1. The results are shown in Figure~\ref{fig_3} and ~\ref{fig_6}. Obviously, the error of PM barely increases with the data size. The reason is that our error only depends on the domain size of attributes in queries, which does not change much by the scale of SSB data. On the other hand, the behavior of R2T is more complicated. For $Q_{c2}$ and $Q_{c3}$, its error decreases first but then increases later; for $Q_{c4}$, its error increases first but then decreases later. The reason is that R2T needs to choose the optimal result based on the truncation threshold, which is closely related to the scale of the database instance. The utility of LS linearly increase with the data size as expected. In terms of running time, all mechanisms linearly increases with the data size, among which the increment of PM is smaller. Compared with the basic solutions, DP-starJ shows superior performance in various star-join query types.

\emph{The impact of Domain size.} To further evaluate the impact of domain size on PM, we extended the star-join count query on the SSB dataset, and set up five queries with different domain value combinations involving two dimension tables. The results are shown in Figure~\ref{fig_10}. Due to the increase in noise with the domain size, the error of the PM will experience a slight increase. When PM perturbs the predicate, its perturbation result is still within the domain value range, which weakens the impact of noise on the results to a certain extent. In addition, the error of PM is still orders of magnitude smaller than R2T and LS.

\emph{Different distributions.} As shown in Figures~\ref{fig_3} and ~\ref{fig_6}, the performance of PM has differences between count and sum queries. In order to further investigate the reasons, we construct data instances following different distributions based on the SSB dataset. The results are shown in Figure~\ref{fig_9}. Firstly, the PM performs best on Uniform distribution, and the error gradually increases as the data distribution becomes more skewed (\eg, Exponential and Gamma distributions). Secondly, for count queries, the error growth rate is higher. Lastly, with increasing data volume, the error of the PM decreases after an initial increase for sum queries. The main reason for this difference is that the result of sum queries depends on the values of the data itself, while the result of count queries depends on the data distribution. To further justify the impact of skewed data on the PM, extensive experiments are conducted by using data following a mixture of Gaussian distributions with different parameters. The results are shown in Figure~\ref{fig_12}. It is obvious that PM has a greater impact on count queries on skewed data. This observation partially suggests that count query results are more dependent on the data distribution.

\emph{Dependency on $GS_Q$.} Our last set of experiments examines the effect $GS_Q$ brings to the utilities of PM, R2T, and LS. We conduct experiments using counting queries with different values $GS_Q$. The results are shown in Figure~\ref{fig_4}. It is observed that PM is insensitive with $GS_Q$ as $GS_Q$ of PM is only related to the queries. When $GS_Q$ increases, the errors of R2T and LS increase rapidly.

\noindent\emph{\textbf{Evaluation on snowflake query.}}
o illustrate the effect of PM on snowflake queries, we select two queries from the TPC-H benchmark, referred to as $Q_{tc}$ and $Q_{ts}$, which are count and sum queries, respectively. The results are shown in Figure~\ref{fig_11}, it is observed that PM outperforms both R2T and LS.

\section{Conclusions}
In this paper, we have presented a novel solution to answer star-join query under differential privacy. We have proposed the definitions of neighboring database instances in different cases of star-join, taking into account the non-trivial number of foreign key constraints. Inspired by the latest output mechanism framework, we have proposed DP-starJ under \textsc{dp} for answering
star-join queries, in which we have designed a new mechanism using predicate perturbation to achieve reasonable utility, efficiency, and scalability.


\begin{acks}
This work was supported by the National Natural Science Foundation of China 61972309, 62272369, 62206207.
\end{acks}

\balance
\bibliographystyle{ACM-Reference-Format}
\bibliography{refer}

\appendix
\section{List of queries and their domain sizes}
We provide the detailed queires, predicates and their corresponding domain sizes on SSB queries and $k$-star queries in this section.

\subsection{The SSB queries} 

$Q_{c1}: 7.$ The domain size of predicate ${Date}.{year}$ is 7.
\begin{lstlisting}[language=SQL,keywordstyle=\color{blue},mathescape,basicstyle=\ttfamily,escapeinside=\{\},showstringspaces=false]
  SELECT count($*$) FROM {Date}, Lineorder 
  WHERE Lineorder.orderdate = {Date}.DK
    AND {Date}.{year} = 1993;  
\end{lstlisting} 

$Q_{c2}: 25 \times 5$, which means that the domain sizes of predicates $Part.category$ and $Supplier.region$ are 25 and 5, respectively.
\begin{lstlisting}[language=SQL,keywordstyle=\color{blue},mathescape,basicstyle=\ttfamily,escapeinside=\{\},showstringspaces=false]
  SELECT count($*$) 
  FROM {Date}, Lineorder, Part, Supplier 
  WHERE Lineroder.SK = Supplier.SK
    AND Lineroder.PK = Part.PK 
    AND Lineorder.orderdate = {Date}.DK
    AND Part.category = 'MFGR#12'
    AND Supplier.region = 'AMERICA';  
\end{lstlisting} 

$Q_{c3}: 5 \times 5 \times 7$.
\begin{lstlisting}[language=SQL,keywordstyle=\color{blue},mathescape,basicstyle=\ttfamily,escapeinside=\{\},showstringspaces=false]
  SELECT count($*$) 
  FROM {Date}, Lineorder, Customer, Supplier 
  WHERE Lineroder.SK = Supplier.SK
    AND Lineroder.CK = Customer.CK 
    AND Lineorder.orderdate = {Date}.DK
    AND Customer.region = 'ASIA'
    AND Supplier.region = 'ASIA'
    AND {Date}.{year} between 1992 and 1997;  
\end{lstlisting} 

$Q_{c4}: 5 \times 25 \times 7 \times 5$.  
\begin{lstlisting}[language=SQL,keywordstyle=\color{blue},mathescape,basicstyle=\ttfamily,escapeinside=\{\},showstringspaces=false]
  SELECT count($*$) 
  FROM {Date}, Lineorder, Customer, Part, Supplier 
  WHERE Lineroder.SK = Supplier.SK
    AND Lineroder.PK = Part.PK 
    AND Lineroder.CK = Customer.CK 
    AND Lineorder.orderdate = {Date}.DK
    AND Customer.region = 'AMERICA'
    AND Supplier.nation = 'UNITED STATES'
    AND {Date}.{year} between 1997 and 1998
    AND Part.mfgr = 'MFGR#1' 
    OR Part.mfgr = 'MFGR#2';  
\end{lstlisting}

$Q_{s2}: 25 \times 5$. 
\begin{lstlisting}[language=SQL,keywordstyle=\color{blue},mathescape,basicstyle=\ttfamily,escapeinside=\{\},showstringspaces=false]
  SELECT sum(Lineorder.revenue) 
  FROM {Date}, Lineorder, Part, Supplier 
  WHERE Lineroder.SK = Supplier.SK
    AND Lineroder.PK = Part.PK 
    AND Lineorder.orderdate = {Date}.DK
    AND Part.category = 'MFGR#12'
    AND Supplier.region = 'AMERICA';  
\end{lstlisting}

$Q_{s3}: 5 \times 5 \times 7$. 
\begin{lstlisting}[language=SQL,keywordstyle=\color{blue},mathescape,basicstyle=\ttfamily,escapeinside=\{\},showstringspaces=false]
  SELECT sum(Lineorder.revenue) 
  FROM {Date}, Lineorder, Customer, Supplier 
  WHERE Lineroder.SK = Supplier.SK
    AND Lineroder.CK = Customer.CK 
    AND Lineorder.orderdate = {Date}.DK
    AND Customer.region = 'ASIA'
    AND Supplier.region = 'ASIA'
    AND {Date}.{year} between 1992 and 1997;  
\end{lstlisting} 

$Q_{s4}: 5 \times 25 \times 7 \times 5$.  
\begin{lstlisting}[language=SQL,keywordstyle=\color{blue},mathescape,basicstyle=\ttfamily,escapeinside=\{\},showstringspaces=false]
  SELECT sum(Lineorder.revenue) 
  FROM {Date}, Lineorder, Customer, Part, Supplier 
  WHERE Lineroder.SK = Supplier.SK
    AND Lineroder.PK = Part.PK 
    AND Lineroder.CK = Customer.CK 
    AND Lineorder.orderdate = {Date}.DK
    AND Customer.region = 'AMERICA'
    AND Supplier.nation = 'UNITED STATES'
    AND {Date}.{year} between 1997 and 1998
    AND Part.mfgr = 'MFGR#1' 
    OR Part.mfgr = 'MFGR#2';  
\end{lstlisting}

$Q_{g2}: 25 \times 5$. 
\begin{lstlisting}[language=SQL,keywordstyle=\color{blue},mathescape,basicstyle=\ttfamily,escapeinside=\{\},showstringspaces=false]
  SELECT sum(Lineorder.revenue), {Date}.{year}, Part.brand 
  FROM {Date}, Lineorder, Part, Supplier 
  WHERE Lineroder.SK = Supplier.SK
    AND Lineroder.PK = Part.PK 
    AND Lineorder.orderdate = {Date}.DK
    AND Part.category = 'MFGR#12'
    AND Supplier.region = 'AMERICA'
  Group by {Date}.{year}, Part.brand
  Order by {Date}.{year}, Part.brand;  
\end{lstlisting} 

$Q_{g4}: 5 \times 25 \times 7 \times 5$.
\begin{lstlisting}[language=SQL,keywordstyle=\color{blue},mathescape,basicstyle=\ttfamily,escapeinside=\{\},showstringspaces=false]
  SELECT sum(Lineorder.revenue - Lineorder.supplycost),
    {Date}.{year}, Part.category
  FROM {Date}, Lineorder, Customer, Part, Supplier 
  WHERE Lineroder.SK = Supplier.SK
    AND Lineroder.PK = Part.PK 
    AND Lineroder.CK = Customer.CK 
    AND Lineorder.orderdate = {Date}.DK
    AND Customer.region = 'AMERICA'
    AND Supplier.nation = 'UNITED STATES'
    AND {Date}.{year} between 1997 and 1998
    AND Part.mfgr = 'MFGR#1' OR Part.mfgr = 'MFGR#2'
  Group by {Date}.{year}, Part.category
  Order by {Date}.{year}, Part.category;  
\end{lstlisting}

\subsection{$k$-star queries}
The $k$-star queries on \textbf{Deezer} and \textbf{Amazon} datasets as follows: The $k$-star counting queries predicate refers to its node range, so the domain size is its number of vertices.

\textbf{Deezer}: the domain size of $k$-star queries is 144000. 

$Q_{2\ast}: $
\begin{lstlisting}[language=SQL,keywordstyle=\color{blue},mathescape,basicstyle=\ttfamily,escapeinside=\{\},showstringspaces=false]
  SELECT count($*$) 
  FROM Edge AS R1, Edge AS R2
  WHERE R1.from_id = R2.from_id 
    AND R1.to_id < R2.to_id 
    AND R1.from_id between 1 and 144000;
\end{lstlisting}

$Q_{3\ast}: $
\begin{lstlisting}[language=SQL,keywordstyle=\color{blue},mathescape,basicstyle=\ttfamily,escapeinside=\{\},showstringspaces=false]
  SELECT count($*$) 
  FROM Edge AS R1, Edge AS R2, Edge AS R3 
  WHERE R1.from_id = R2.from_id 
    AND R1.from_id = R3.from_id 
    AND R1.to_id < R2.to_id 
    AND R2.to_id < R3.to_id 
    AND R1.from_id between 1 and 144000
    AND R3.from_id between 1 and 144000;
\end{lstlisting}

\textbf{Amazon}: the domain size of $k$-star queries is 335000. 

$Q_{2\ast}: $
\begin{lstlisting}[language=SQL,keywordstyle=\color{blue},mathescape,basicstyle=\ttfamily,escapeinside=\{\},showstringspaces=false]
  SELECT count($*$) 
  FROM Edge AS R1, Edge AS R2 
  WHERE R1.from_id = R2.from_id 
    AND R1.to_id < R2.to_id 
    AND R1.from_id between 1 and 335000;
\end{lstlisting}

$Q_{3\ast}: $
\begin{lstlisting}[language=SQL,keywordstyle=\color{blue},mathescape,basicstyle=\ttfamily,escapeinside=\{\},showstringspaces=false]
  SELECT count($*$) 
  FROM Edge AS R1, Edge AS R2, Edge AS R3 
  WHERE R1.from_id = R2.from_id 
    AND R1.from_id = R3.from_id 
    AND R1.to_id < R2.to_id 
    AND R2.to_id < R3.to_id 
    AND R1.from_id between 1 and 335000
    AND R3.from_id between 1 and 335000;
\end{lstlisting}

\end{document}